\documentclass[12pt]{article}
\usepackage{graphicx} 
\usepackage{amsmath, amssymb, amsthm} 
\usepackage{fullpage, setspace,comment, color}
\usepackage[dvipsnames]{xcolor}
\usepackage[normalem]{ulem}
\usepackage[linesnumbered,ruled,vlined]{algorithm2e}

\usepackage{hyperref}
\usepackage{subcaption}
\usepackage[shortlabels]{enumitem} 
\usepackage{booktabs} 
\RequirePackage[authoryear]{natbib}
\usepackage{url}
\newtheorem{theorem}{Theorem}  
\newtheorem{remark}{Remark}
\newtheorem{definition}{Definition}
\DeclareMathOperator*{\argmin}{arg\,min}
\DeclareMathOperator*{\argmax}{arg\,max}

\usepackage{float}

\usepackage{multirow}
\usepackage{array}

\title{ 
FDR Control via Neural Networks under Covariate-Dependent Symmetric Nulls}
\author{Taehyoung Kim, Seohwa Hwang\thanks{Email: westshine28@snu.ac.kr}  and Junyong Park\\
{\normalsize \it Department of Statistics, Seoul National 
University, Seoul, Korea}
}

\date{}

\begin{document}
\maketitle
\begin{abstract}
{
In modern multiple hypothesis testing, the availability of additional covariate information alongside the primary test statistics has motivated the development of more powerful and adaptive inference methods. However, most existing approaches rely on $p$-values that are precomputed under the assumption that their null distributions are independent of the covariates. In this paper, we propose a novel framework that first derives covariate-adaptive $p$-values  from 
the assumption of symmetric null distribution  of the primary variable given the covariates, without imposing any parametric assumptions. Building on these data-driven $p$-values, we then employ a neural network model to learn a covariate-adaptive rejection threshold via the mirror estimation principle, optimizing the number of discoveries while maintaining valid false discovery rate (FDR) control.
Furthermore, our estimation of the conditional null distribution enables the computation of 
p-values directly from the raw data.
The proposed method not only provides a principled way to derive covariate-adjusted 
p-values from raw data but also allows seamless integration with previously established p-value–based methods.
Simulation studies demonstrate that the proposed method outperforms existing approaches in terms of  powers. We further illustrate its applicability through two real data analyses—age-specific blood pressure data and U.S. air pollution data.}

\end{abstract}

\section{Introduction}

Multiple hypothesis testing has become a cornerstone of modern statistical inference, especially in large-scale scientific studies where thousands of hypotheses are examined simultaneously. The control of the false discovery rate (FDR) has emerged as a standard criterion to balance the trade-off between discovering true signals and limiting false positives. Traditional FDR control procedures, such as the Benjamini–Hochberg (BH) method in \cite{benjamini1995controlling},
are built on the assumption that, under the null hypothesis, the $p$-values follow a uniform or super-uniform distribution in the interval $[0,1]$. This assumption is reasonable when valid $p$-values can be readily computed from a known null distribution that does not depend on any additional variables.

In many modern applications, however, test statistics are accompanied by auxiliary or covariate information that is informative about the underlying structure of the testing problem. Incorporating such information has been shown to improve power while maintaining rigorous FDR control, leading to a growing literature on covariate-adaptive multiple testing methods,  
for example \cite{boca2018regression}, \cite{ignatiadis2016data} and \cite{lei2018adapt}.  
Despite this progress, most existing approaches implicitly assume that valid $p$-values---already adjusted for covariate effects---are available as input. This assumption, however, is often unrealistic. The core challenge arises when the null distribution of the test statistic itself depends on the covariates, making the computation of valid $p$-values a non-trivial task. { As a motivating example, consider blood pressure, which cannot be interpreted uniformly across all age groups. For a young adult, a systolic blood pressure of 140 mmHg can be a significant indicator of an abnormal condition, whereas for an individual in their eighties, a similar reading might simply reflect natural aging. This age-related variation means that the ``null'' distribution of blood pressure inherently shifts with the covariate, and consequently, naively computing $p$-values against a single global null distribution may lead to invalid inference.
}

To address this challenge, we propose a methodology that controls the FDR directly at the raw data level. Our approach is founded on the assumption that the conditional distribution of the data, given the covariate under the null hypothesis, is symmetric. This symmetry allows us to adopt the idea of a ``mirroring property,'' a nonparametric approach for controlling the FDR in \cite{Dai2022FDRDataSplitting}  rather than relying on parametric model assumptions. 
This principle is paired with a ``zero assumption,'' in the spirit of \cite{efron2004large}, which posits that alternative hypotheses are predominantly one-sided. 
The practical implementation of this framework hinges on accurately identifying the center of symmetry, a key challenge since the observed data are a mixture of null and alternative hypotheses. 
Therefore, our primary contribution is a novel, data-dependent procedure to estimate this function of centers. Unlike the original work of \cite{efron2004large}, which assumed a Normal distribution and used a pre-determined interval for estimation without covariates, our method is non-parametric, accommodates covariates, and iteratively filters data points likely generated from the alternative hypothesis. 
Once the function of the center conditioning on the covariate 
is obtained, we derive a p-value from the estimated information and then find a boundary for rejection. Both the p-values and the boundaries are functions of the covariate, where 
the number of rejections is optimized with controlling FDR.

{ The utility of our proposed framework extends beyond its primary function of controlling FDR from raw data. Its applicability is not limited to inherently symmetric distributions; for instance, in our analysis of air pollution data, we demonstrate that our algorithm can be applied to daily PM2.5 concentrations, which are typically right-skewed, after being transformed to satisfy our symmetry assumption. This showcases the method's flexibility in handling a wide range of real-world data. }
Our proposed method to identify the center of the primary variable given the covariate 
provides a principled way to compute covariate-adjusted $p$-values. This feature ensures that our findings are compatible with and can be used with conventional $p$-value–based approaches. We present the complete methodology and its evaluation through simulations and real-data examples.

{
This paper is organized as follows. 
Section~\ref{sec:related} reviews existing covariate-adaptive FDR procedures and discusses their limitations arising from the assumption of given covariate-independent $p$-values. 
Section~\ref{sec:method} introduces our proposed framework, which performs multiple testing at the level of raw data under conditional symmetry and zero assumptions. 
We describe the algorithm for estimating the conditional center function  through iterative trimming and the subsequent computation of empirical $p$-values based on the estimated null distribution. The neural-network–based optimization procedure for learning a covariate-adaptive rejection threshold is also detailed. 
Section~\ref{sec:simulation} presents extensive simulation studies that compare our method with existing procedures, which are developed for p-values, in various covariate-dependent data-generating settings, demonstrating superior power and accurate FDR control. 
Section~\ref{sec:realdata} provides two real-data applications—the age-specific blood-pressure data and U.S.\ air-pollution data—which illustrate the flexibility and practical usefulness of the proposed approach. 
Finally, Section~\ref{sec:conclusion} concludes with a discussion of implications, possible extensions, and directions for future research.}

\section{Related Work}
\label{sec:related}
The classical false discovery rate (FDR) control procedure, introduced by \citet{benjamini1995controlling}, treats all hypotheses symmetrically, disregarding any additional information that may be available for each hypothesis. However, in many large-scale testing problems, hypotheses are often accompanied by covariates that may be informative about the likelihood of a hypothesis being null or alternative. Incorporating such covariates has been shown to substantially improve detection power while maintaining control over the FDR.

One of the early systematic approaches to considering covariates in FDR control is the Independent Hypothesis Weighting (IHW) method proposed by \citet{ignatiadis2016data}. IHW partitions hypotheses based on covariate information and assigns data-driven weights to each partition, enhancing the detection of true signals while rigorously controlling the FDR. The method assumes that the covariate is independent of the $p$-value under the null hypothesis, a condition that ensures the validity of the FDR control.

{Some other approaches are the AdaPT framework by \citet{lei2018adapt} and AdaFDR framework by \citet{zhang2019fast}.
The AdaPT framework, which provides an interactive, model-free procedure for multiple testing with side information, adaptively estimates rejection thresholds based on covariates during sequential testing and maintains FDR control under mild assumptions. A main advantage of AdaPT is its flexibility, allowing incorporation of complex models, such as generalized additive models (GAMs) or gradient boosted trees, to capture the relationship between covariates and local false discovery rates. The AdaFDR framework, which uses Gaussian mixture models to learn covariate-dependent thresholds, offers a computationally efficient for high-dimensional covariates and guarantees FDR control under independence conditions.}


NeuralFDR by \citet{xia2020neuralfdr} introduced a deep learning approach to FDR control by learning an optimal threshold function from covariates through a neural network. NeuralFDR estimates a mapping from covariates to $p$-value thresholds and guaranties asymptotic FDR control under certain regularity conditions. This method can flexibly capture complex covariate effects and is particularly powerful in high-dimensional or highly nonlinear settings where traditional parametric models may fail.

Despite these advances, most existing methods share a fundamental assumption: under the null hypothesis, the $p$-values are uniformly distributed on $[0,1]$ and independent of covariates. NeuralFDR in \citet{xia2020neuralfdr} relaxes this assumption slightly, requiring only that the conditional distribution of null $p$-values be symmetric about $1/2$. Nevertheless, the underlying paradigm remains that the $p$-values are already computed in a manner that eliminates dependence on covariates.
A limitation of this paradigm is that it implicitly assumes access to the conditional null distribution of the raw data given the covariates, from which valid $p$-values can be derived. However, in many applications, such conditional distributions are unknown and cannot be computed exactly.  
This motivates a shift in perspective: instead of taking $p$-values as primitives, one may directly estimate the joint distribution of the raw data and covariates under the null hypothesis. From this,  we can construct covariate-calibrated test statistics or $p$-values that more faithfully reflect the dependence structure in the data. Unlike existing methods such as IHW, AdaPT, AdaFDR, and NeuralFDR, which operate on precomputed $p$-values, this approach emphasizes estimation of the conditional null distribution as a prerequisite for valid multiple testing with covariates.




\section{Proposed Methodology}
\label{sec:method}
\subsection{Problem Definition}
Suppose that the observed data points are
$(X_i, Y_i), X_i = (x_{i1}, \dots , x_{id}) \in \mathbb{R}^d$ 
for $1\leq i \leq n$
where $d$ is the dimension of the covariate.  
We are interested in testing 
$H_{0i}$ vs. $H_{1i}$ based on $Y_i$ for $1\leq i \leq n$, but 
we have additional information $X_i$ that 
makes the decision boundary different from 
traditional multiple testing without covariates. 
We want to achieve two goals:   
$(i)$ the control of  a Type I error rate such as 
the false discovery rate $(FDR)$ 
 \begin{eqnarray}
 FDR =  E(FDP) =  E\left(\frac{V}{R \vee 1} \right) 
 \label{eqn:FDR}
\end{eqnarray}
 where $FDP$ represents the false discovery proportion $\frac{V}{R \vee 1}$,   
  $R$ is the number of rejections or discoveries among $n$ hypotheses 
 testing and $V$ is the number of false rejections and   
  $(ii)$   maximization of the number of rejections $R$.  
{More specifically, 
we define the conditional densities 
of $y$ given $X=x$
under the null  
and alternative hypotheses, 
$f_{0,x}(y) \equiv f_0(y|x)$ 
and $f_{1,x}(y)\equiv f_1(y|x)$ and denote $f_x (y)$ as the mixed distribution of 
 $f_{0,x}(y)$ and $f_{1,x}(y)$ 
 with mixing proportions $1-\epsilon_x$ and $\epsilon_x$
 at $x$ :}
\begin{equation}
    f_x (y) = (1-\epsilon_x) f_{0,x}(y) + \epsilon_x f_{1,x}(y).
\label{eqn:modeldef}
\end{equation}

An example 
is the case where  
$p$-values are given 
which has been considered in many studies 
such as \cite{xia2020neuralfdr}, \cite{ignatiadis2016data} and \cite{lei2018adapt} etc. 
In this case,  
most existing studies that assume $Y$ to be a $p$-value use 
 $f_{0,x} (p)= {\rm I}_{[0,1]}(p)$, 
 the uniform distribution, regardless of the covariate $x$. 
 This implies that the covariate is independent of $p$-values under the null hypothesis, i.e., the covariates have been adjusted during the computation of $p$-values.  
However, such an assumption is strong and practically unrealistic, 
since there is very little literature describing how these $p$-values are actually computed from the raw data.  
In many real applications, $p$-values are not available in a covariate-adjusted form, and the relationship between the covariate and the response must be modeled directly from the data. 
Our proposed framework provides a novel contribution by addressing this gap: 
we present multiple testing starting at the level of the raw data $(X_i, Y_i)$, 
rather than relying on precomputed or assumed covariate-adjusted $p$-values. 
This approach allows us to explicitly model and learn the dependence of the test statistic on the covariates, 
resulting in a more principled and data-driven inference procedure that extends beyond the scope of existing $p$-value-based methods.

In this paper, we propose a method that operates directly at the raw data level rather than assuming precomputed $p$-values. 
In this case, the null distribution of $Y$ given covariates $X=x$ may itself depend on these covariates.
Unlike existing methods that impose a specific parametric form (e.g., normality or exponential family assumptions) for $f_{0,x}(y)$, 
we do not assume any parametric structure. 
Instead, we adopt a more relaxed and flexible assumption that the null distribution is \emph{symmetric} about a covariate-dependent center $m(x)$ as follows: 
 \begin{eqnarray}
   f_{0,x}(y) = f_{0,x}(2m(x)-y). 
 \end{eqnarray}
This assumption of symmetry reflects the inherent balance of deviations around $m(x)$ under the null hypothesis, 
while allowing for heterogeneity in scale and shape across different covariate values. 
By working directly with the raw data $(X_i, Y_i)$ in this nonparametric and covariate-dependent framework, 
our proposed procedure extends beyond conventional $p$-value-based multiple testing methods 
and provides a more general and data-adaptive approach to FDR control.  
This modeling flexibility enables us to capture complex relationships between the covariate and response distributions 
without requiring explicit specification of the null model.
Furthermore, we introduce a zero assumption considered in \cite{efron2004large} to avoid the identifiability issue of discriminating $f_{0,x}$ and $f_{1,x}$. 
Considering these aspects, we introduce the following two assumptions.  which are 
compatible with or   
weaker than the traditional requirements. 

{\bf Assumptions:}
\begin{enumerate}
    \item[{\bf A1.}] The conditional null distribution is symmetric around the function of center $m(x)$, i.e., 
    \begin{eqnarray}
      f_{0,x}(y) = h(y-m(x))
      \label{eqn:symmetric}
    \end{eqnarray}
    where $h(y)$ is a density function symmetric around zero. 
    \item [{\bf A2.}] (Zero Assumption) 
    The support of $f_{1,x}(y)$ is     
  $ (a(x), \infty)$ for some function   $a(x) > m(x)$.  
\end{enumerate}
\bigskip 

A main contribution of our work is 
the estimation of the center of the conditional null distribution, $m(x)$. 
{A relevant study is the estimation of 
the empirical null distribution. 
In \cite{efron2004large}, the estimation of the empirical null distribution provides a way to identify the center of the null distribution, but this is done under the assumption that the empirical null is normal and without considering covariates. In our setting, however, the covariate must be incorporated in estimating the center of the conditional null distribution, and we avoid assuming normality or any other parametric model. While the underlying idea is inherited from the conditional maximum likelihood approach in Efron (2004), we cannot directly apply MLE since no parametric form for the null distribution is imposed. Moreover, unlike \cite{efron2004large}, which relies on a predetermined interval that includes the data generated only from the null hypothesis 
to construct the conditional likelihood, our approach adaptively selects the interval to better reflect the data structure.
}The choice of the interval  
presents a significant challenge because our data consist of both null and alternative cases, yet the estimation requires isolating and utilizing only the subset of data generated under the null hypothesis. Effectively distinguishing and leveraging these null-generated data is crucial for accurate estimation and, consequently, for controlling the false discovery rate. 
Furthermore, depending on how we characterize the symmetric property of $f_{0,x}$ around $m(x)$, 
we propose an algorithm to estimate $m(x)$ in the following section. 

\subsection{Proposed Algorithm for Estimating $m(x)$}
In our proposed method, a baseline idea is data trimming, which involves determining how to effectively trim the data to select those generated from the conditional null distribution of $y$ given $X$, or, more precisely, the neighborhood of $X$.
One main difficulty in estimating the center $m(x)$ is that the data consists of those derived from the null and alternative hypotheses, while we should only use the data from the null hypotheses.

Figure \ref{fig:alg1justify} compares two estimation approaches for the center function  $m(x)$ : one that estimates $m(x)$  without separating data under the null and alternative hypotheses, and our proposed method 
(which will be shown later) that estimates $m(x)$ after trimming observations from both sides appropriately.
It is evident that if the data from the alternative hypothesis are not removed, the estimated center becomes biased toward the regions where alternative data are present.
This bias distorts the estimation of the true center under the null hypothesis and consequently affects the overall results of multiple testing conservatively.
Our trimming-based approach mitigates this issue by reducing the influence of alternative samples in estimating  $m(x)$.

\begin{figure}[ht]
    \centering
    \includegraphics[width=0.95\linewidth]{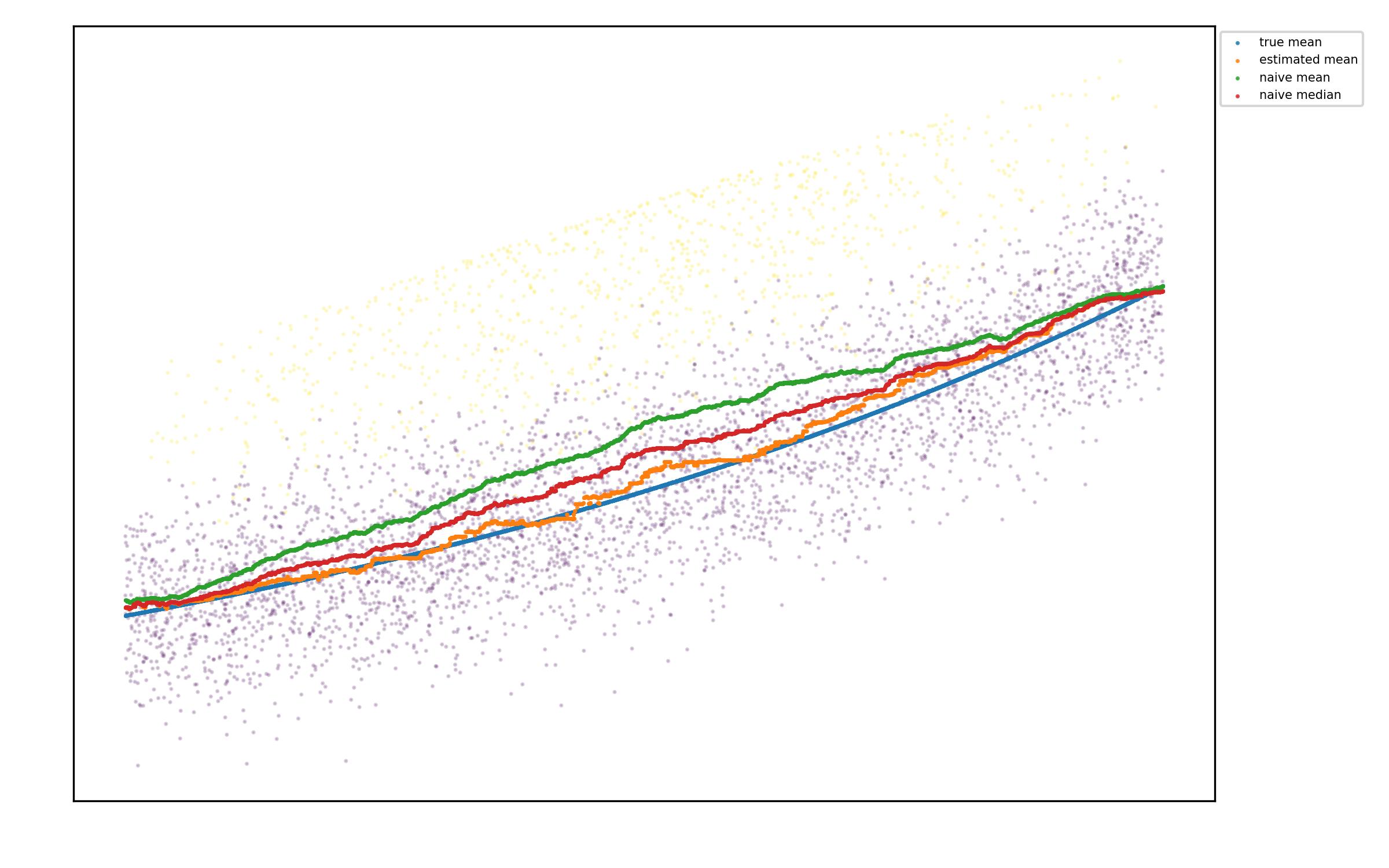}
    \caption{naive estimation of the mean(median) vs proposed method}
    \label{fig:alg1justify}
\end{figure}

Our proposed method for the estimation of $m(x)$ in $x = X_i$ is based on nearby data points of $(X_i, y_i)$. We first start by selecting the local neighborhood based on the available data as follows. For given $i$ and $\delta>0$, define
\begin{equation}
    {\cal N}_{\delta} (X_i) =\{j : |X_i - X_j| < \delta\}
\end{equation}
where $\delta$ is any given  bandwidth. 
In addition, we also define the collection of $Y_j$'s for
$j \in {\cal N}_{\delta}(X_i)$ which is
\begin{equation}
{\cal Y}_{\delta} (X_i) 
=\{Y_j : j \in {\cal N}_{\delta}(X_i)\}. 
\end{equation}

Among the primary values of $Y_j \in {\cal Y}_{\delta}(X_i)$, those who follow the null distribution should be approximately symmetric around $m(X_i)$ for some sufficiently small $\delta>0$.
Since $f_{0,x}$ is assumed to be symmetric around $m(x)$,
we need to quantify this symmetry based on data $Y_j \in {\cal Y}_{\delta} (X_i)$.  
As a method to estimate the center solely from the null distribution, we consider cases where the difference between the sample mean and the sample median remains below a certain acceptable level.
Since the closer the data within the selected interval are to being symmetric, the closer these two statistics will be, we propose an algorithm that iteratively removes tail values until the two statistics are sufficiently close. This process is expected to effectively remove values from alternative distributions and null distributions of the corresponding region.
Using this criterion, we truncate data that meet a suitable condition under the null hypothesis and propose an algorithm that estimates $m(x)$ based on the center of the truncated data.
We provide more details of our proposed algorithm in the following paragraphs and {\bf Algorithm}~\ref{alg:alg1}.

For each fixed $x$, to test the symmetry, we calculate the test statistic $T$, which was introduced in \citet{miao2006sym},
\begin{equation}
    T = \frac{\hat{\mu}(x) - \hat{\nu}(x)}{\hat{d}(x)}
\label{TS}
\end{equation}
where 
\begin{equation}
    \hat{\mu}(x) = \frac{\sum_{j \in {\cal N}_\delta(x)} Y_j}{n},
\label{eq:mu}
\end{equation}
\begin{equation}
    \hat{\nu}(x) = \argmin_{m} \sum_{j \in {\cal N}_\delta(x)} |y_j - m|,
\label{eq:nu}
\end{equation}
\begin{equation}
    \hat{d}(x) = \sqrt{\frac{\pi}{2}} \\ \frac{\sum_{j \in {\cal N}_\delta(x)} |Y_j - \hat{\nu}(x)|}{n}.
\label{eq:mad}
\end{equation}
In other words, 
$\hat{\mu}(x)$ and $\hat{\nu}(x)$ are the sample mean and  
the sample median of $Y_j$'s in  ${\cal N}_{\delta}(x)$, respectively, 
and  similarly, the estimate of a dispersion measure $\hat{d}(x)$  is the mean absolute deviation around the median of $y_j$'s in  ${\cal N}_{\delta}(x)$.


The motivation for the algorithm comes from Corollary 1 of \citet{miao2006sym}. The asymptotic distribution of $T$ is given as
\begin{equation}
    \sqrt{n} \ T \xrightarrow{d} N(0, \sigma_T^2)
\end{equation}
where
\begin{equation}
    \sigma_T^2 = \frac{2}{\pi\tau^2}\left(\sigma^2 + \frac{1}{4f(\nu)^2} - \frac{\tau}{f(\nu)}\right),
\end{equation}
if the underlying distribution is symmetric under certain regularity conditions. {See {\bf Theorem}~\ref{thm:cond} for more details. }

Next, as suggested in \citet{miao2006sym}, using the kernel density estimator, we estimate $\sigma_T^2$ as
\begin{equation}
    \hat{\sigma}_T^2 = \frac{2}{\pi\hat{\tau}^2(x)}\left(\hat{\sigma}^2(x) + \frac{1}{4\hat{f}^2(\hat{\nu}(x))} - \frac{\hat{\tau}(x)}{\hat{f}(\hat{\nu}(x))} \right)
\label{eq:sigmaT}
\end{equation}
where
\begin{align}
    &\hat{\sigma}^2(x) = \frac{1}{|{\cal N}_\delta(x)|} \sum_{j \in {\cal N}_\delta(x)} (Y_j - \hat{\mu}(x))^2, \\
    &\hat{\gamma}(x) = \frac{1}{|{\cal N}_\delta(x)|} \sum_{j \in {\cal N}_\delta(x)} Y_j I(Y_j < \hat{\nu}(x)),\\
    &\hat{\tau}(x) = \hat{\mu}(x) - 2\hat{\gamma}(x), \\
    & \hat{f}(y) = \frac{1}{|{\cal N}_\delta(x)|}\sum_{j \in {\cal N}_\delta(x)} K_{h_y}(y - Y_j).
\end{align}

We propose to use the test statistic
$T$ as a criterion for data trimming 
 used in the {\bf Algorithm}\ref{alg:alg1}, determining which observations should be excluded in order to estimate the center function 
$m(x)$ based primarily on data generated under the null hypothesis as follows. If $T>0$, or equivalently, $\hat{\mu}(x)$ is larger than $\hat{\nu}(x)$, it implies that the distribution is skewed to the right, indicating that there exist outliers in the right tail. Thus, we remove the largest value among ${\cal N}_\delta(x)$. Similarly, $T$ being significantly smaller than 0 would imply that there exist outliers on the left tail. 
As a remedy for the heavy computation, for the iterative process is to be done for each data point, we may exclude several data points at once. As shown in Figure \ref{fig:rm1}, \ref{fig:rm5} below, such a method does not show significant difference in the result; nevertheless, it reduces the computation time meaningfully by reducing the total number of iterations.
\begin{figure}[ht]
    \centering
    \begin{minipage}{0.45\linewidth}
        \centering
        \includegraphics[width=\linewidth]{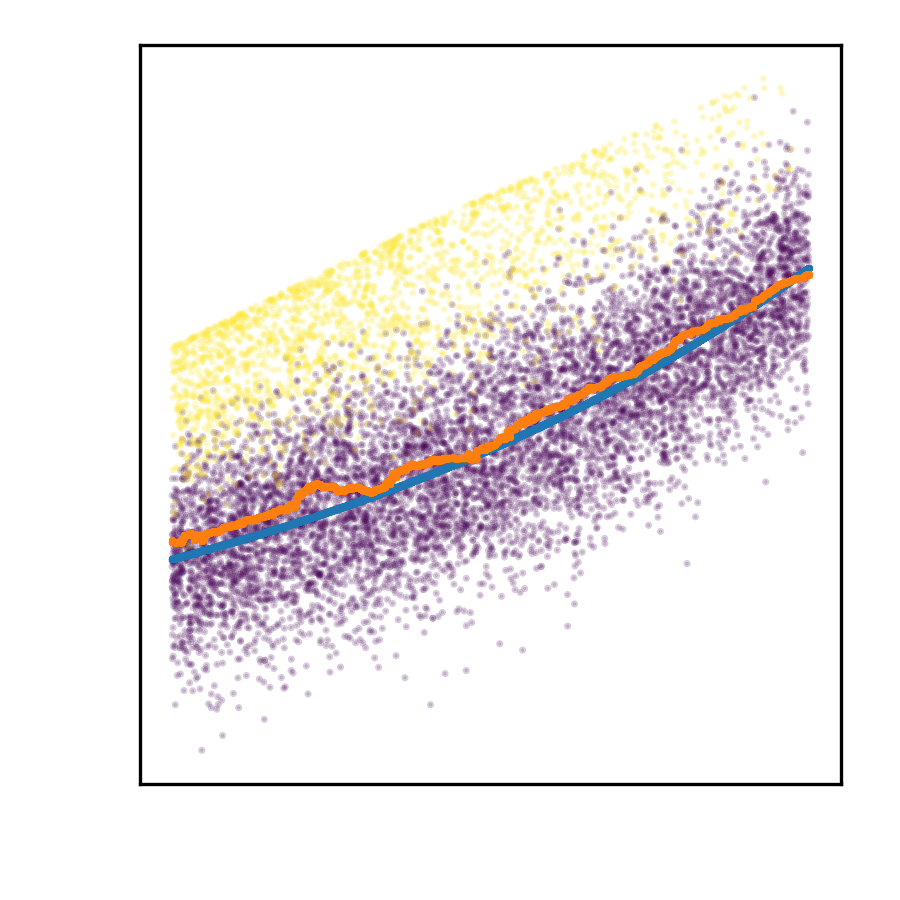}
        \caption{Remove one value at once}
        \label{fig:rm1}
    \end{minipage}\hfill
    \begin{minipage}{0.45\linewidth}
        \centering
        \includegraphics[width=\linewidth]{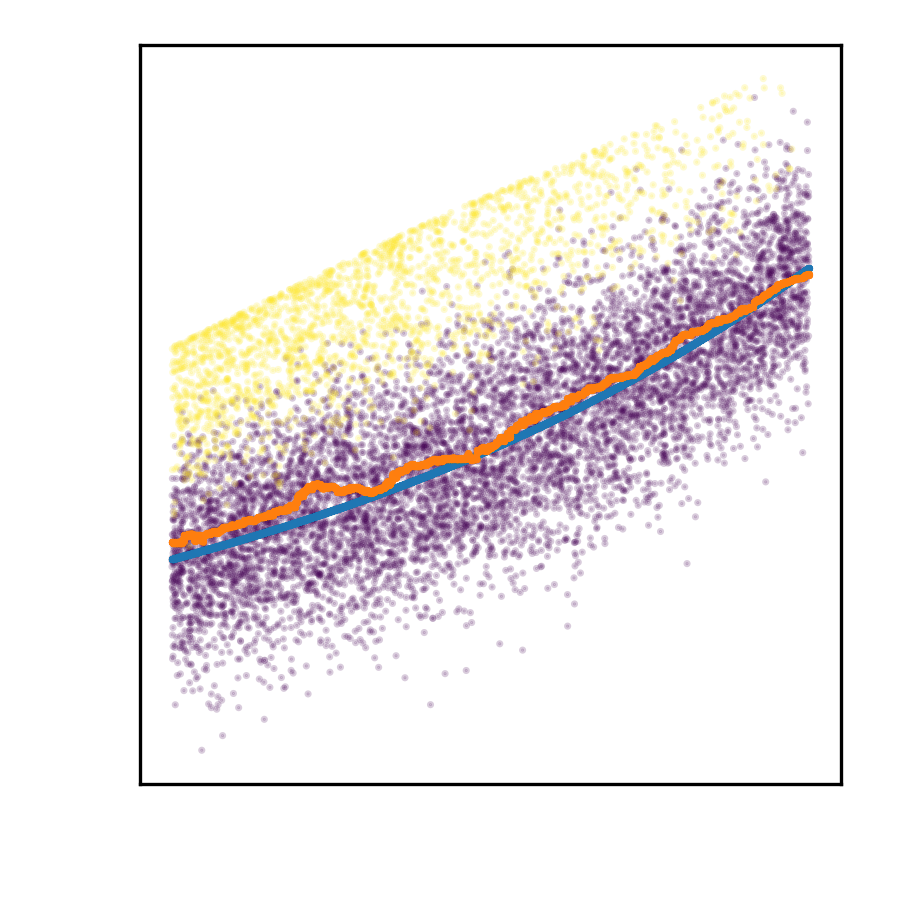}
        \caption{Remove five value at once}
        \label{fig:rm5}
    \end{minipage}
\end{figure}

This procedure is repeated until $T$ is close enough to zero. When the procedure is completed, we expect the selected data $y_j$
to approximate a symmetric distribution, specifically a symmetric truncation of the null distribution. 
We then take the median of the selected samples as the estimate of the conditional mean $\hat m(x)$ and save the largest value among those as an initial value for the threshold, since the region above the threshold is expected to contain all the alternative hypotheses.
\begin{remark}
Since the test statistic is robust and thus may still allow a few data from the alternative, we took the median instead of the mean, but in practice there were no significant differences in the result.
\end{remark}

For a theoretical perspective, the {\bf Algorithm} \ref{alg:alg1} uses a similar idea with backward elimination in variable selection. If we write the elements of ${\cal Y}_{\delta}(X_i)$ as $Z_{(1)} < \cdots<Z_{(N)}$, then the step with selected data $Z_{(i)}, \ l\leq i \leq u$ can be considered as hypothesis testing for the following null hypothesis:
``There exist $L$ and $U$  such that $Z_{(l-1)} < L \leq Z_{(l)}, \ Z_{(u)} \leq U < Z_{(u+1)}$, and $f_{0,x}$ truncated on $[L, U]$ is symmetric." A theoretical analysis of the {\bf Algorithm} \ref{alg:alg1} is given in Section \ref{sec:theory}.

\begin{algorithm}[ht]
    \caption{Estimating Conditional Mean of Null Distribution}\label{alg:alg1}
    
    \KwIn{Dataset $\left\{X_j,Y_j\right\}_{j=1}^n$, Bin width $\delta$}
    \KwOut{ $\left\{m(X_j)\right\}_{j=1}^n$, $\left\{t_0\left(X_j\right)\right\}_{j=1}^n$}

    \For{$j \in \left\{1,2,\ldots,n\right\}$}{
        Set $\mathcal{N}_\delta\left(X_j\right)\gets\left\{ i : |X_i - X_j| \leq \delta\right\}$\\
        Calculate $\hat{\mu}(X_j), \hat{\nu}(X_j), \hat{\sigma}(X_j)$ and $\hat{\sigma}_T\left(X_j\right)$ with $\mathcal{N}_\delta(X_j)$\\
        Calculate $T\gets \frac{\hat{\mu}\left(X_j\right) -\hat{\nu}\left(X_j\right)}{\hat{\sigma}\left(X_j\right)}$\\
        \While{ $|T |>1.96\hat{\sigma}_T(X_j)$}{
            \eIf{$T > 1.96\hat{\sigma}_T(X_j)$}{
                Set $i^* \gets \arg\max_{i\in\mathcal{N}_{\delta}(X_j)} Y_j$\\
            }{
                Set $i^* \gets \arg\min_{i\in\mathcal{N}_{\delta}(X_j)} Y_j$\\
            }
            Update $\mathcal{N}_\delta(X_j)\gets \mathcal{N}_\delta(X_j)\backslash\left\{i^*\right\}$\\
            Update $\hat{\mu}(X_j), \hat{\nu}(X_j), \hat{\sigma}(X_j),\hat{\sigma}_T\left(X_j\right)$ and $T$ with updated $\mathcal{N}_\delta(X_j)$\\
        }
        Set $m(X_j) \gets \hat{\mu}(X_j)$ and $t_0(X_j)\gets \max\left\{Y_i : i \in \mathcal{N}_\delta(X_j)\right\}$
    }
    \KwRet{  $\left\{m(X_j)\right\}_{j=1}^n$, $\left\{t_0\left(X_j\right)\right\}_{j=1}^n$}
\end{algorithm}

\subsection{Deriving the Optimal Threshold}
After estimating the null mean function $\hat m(x)$ using the {\bf Algorithm} \ref{alg:alg1}, 
we need to decide a critical value for a given $X=x$ 
which is called the threshold function denoted by ${t}(x)$. 
Instead of using the primary variable of interest, we now derive a p-value for each observations.

First, to calculate the empirical distribution, we take only the samples in ${\cal N}_{\delta} (X_i)$ that are smaller than the estimated mean $\hat{m}(x)$ and also mirror them 
to the estimated mean $\hat m(X_i)$.  
With $Y_j \leq \hat m(X_i)$ and their corresponding mirrored data $2\hat m(X_i) - Y_j$ about $\hat m(X_i)$, 
we define  
\begin{eqnarray*}
{\cal Y}_{\delta,0}(X_i)&\stackrel{def}{\equiv} &
\{Y_j : Y_j \leq \hat m(X_i), j\in {\cal N}_\delta(X_i) \}
\cup \{2\hat m(X_i)-Y_j : Y_j \leq \hat m(X_i), j\in {\cal N}_\delta(X_i) \}\\
&\stackrel{def}{\equiv}&  A(X_i) \cup {\cal M}(A(X_i)) 
\end{eqnarray*}
where ${\cal M}(A(X_i))$ is the mirror set or the reflected set of $A(X_i)$
about $\hat m(X_i)$. 
By mirroring  $Y_j \in A(X_i)$ to $\hat m(X_i)$, 
we obtain an approximate null distribution of $Y_j$ given  $X=X_i$ 
for the part of $y\geq \hat m(X_i)$.

 Then, we calculate a corresponding p-value based on the empirical distribution, which is reasonable due to the assumption of a symmetric null distribution in  \eqref{eqn:symmetric}: 
 {for $Y_i$ corresponding to $X_i$}, 
\begin{eqnarray}
    p_{i} 
    &=&  \frac{1}{|{\cal Y}_{\delta,0}(X_i)|} \sum_{Y_j \in {\cal Y}_{\delta,0}(X_i)} I( Y_j > Y_i ).
\end{eqnarray}

We now estimate $t(x)$, the threshold function which maximizes
the number of rejections while controlling the false discovery proportion ($FDP$) 
which is defined  in \eqref{eqn:FDR}. 
For a given threshold function $\hat{t}(x)$, 
we define the notation which is crucial in multiple testing procedures as follows: 

\begin{eqnarray}
R(\hat{t}) &=& \sum_{i=1}^{n} I(p_i < \hat{t}(X_i)), 
\label{eqn:R}\\
V(\hat{t}) &=& \sum_{i=1}^{n} I(p_i < \hat{t}(X_i), \ H_i = 0),\\
\hat{V}(\hat{t}) &=& \sum_{i=1}^{n} I(p_i > 1 - \hat{t}(X_i)), \label{eqn:hatV} \\
\widehat{FDP}(\hat{t}) &=& \frac{\hat{V}(\hat{t})}{R(\hat{t})}.
\label{eqn:hatFDP}
\end{eqnarray}

$R(t)$ represents the number of rejections ($R$), which is observable, while $V(t)$ is the number of false rejections ($V$), which is not directly observed since 
the true hypothesis is unknown. 
In this stage, since we need to estimate $V(t)$ for a given threshold function $t$, we use the mirroring property, assuming that the p-values follow the uniform distribution $Unif[0,1]$. 
$V(t)$ can be estimated as in \eqref{eqn:hatV}, 
which leads to the estimated false discovery proportion (FDP) in \eqref{eqn:hatFDP}. 
The remaining issue is how we estimate the threshold function $t(x)$ that satisfies some desired properties in multiple testing. 
For an estimated $\hat m(x)$, 
our goal is to estimate the threshold function 
by solving the following optimization problem:
\begin{eqnarray}
    \max_{\hat{t}} R(\hat{t}) \quad \mbox{subject to} \quad \widehat{FDP}(\hat{t}) \leq \alpha.
\end{eqnarray}

We use deep learning to optimize the threshold function that maximizes the number of rejections while controlling the estimated false discovery proportion under a given nominal level $\alpha$. During this procedure, two techniques are used for efficient learning. First, to prevent non-vanishing gradient problems, the sigmoid function 
$\sigma(x)$ is used instead of the indicator function when calculating the rejections and the false rejections. Second, 
for $t_0(x)$ obtained in the {\bf Algorithm} \ref{alg:alg1}, 
$(X_i, t_0(X_i))$ values are used as initial values, so there exists a short pre-training step for initialization. These two techniques make training stable and computationally efficient.

Writing the above in equations, we optimize the threshold function $\hat{t}(x)$ to maximize
\begin{equation}
    R_\sigma(\hat{t}) = \sum_{i=1}^{n} \sigma (\lambda  (\hat{t}(X_i) - p_i))
\label{eq:Rsigma}
\end{equation}
instead of
   $ R(\hat{t}) = \sum_{i=1}^{n} I(p_i < \hat{t}(X_i))$
while the proportion of false rejection $FDP_\sigma(\hat{t})$, derived from $R_\sigma(\hat{t})$ and an estimate of false rejections $V_\sigma(\hat{t})$, is controlled:
\begin{align}
    & FDP_\sigma(\hat{t}) = \frac{V_\sigma(\hat{t})}{R_\sigma(\hat{t})} \leq \alpha,
    \label{eq:FDPsigma}
    \\
    & V_\sigma(\hat{t}) = \sum_{i=1}^{n} \sigma (\lambda (p_i - (1 - \hat{t}(X_i)) )).
    \label{eq:Vsigma}
\end{align}
Here, $\lambda$ is a constant that is arbitrarily chosen. Too large $\lambda$ could lead to vanishing gradient problems, while too small $\lambda$ can enlarge the difference between $R$ and $R_\sigma$, $V$ and $V_\sigma$.

To optimize the threshold function, we use a multilayer perceptron(MLP), which is defined recursively as in the following:
\begin{align*}
    & x^{(0)} = x, && x\in \mathbb{R}^{d} \\
    & x^{(1)} = \phi(b^{(1)} + (w^{(1)})^T x^{(0)}), && b^{(1)} \in \mathbb{R}^{d_{h}}, &&w^{(1)} \in \mathbb{R}^{d \times d_{h}} \\
    & x^{(k)} = \phi(b^{(k)} + (w^{(k)})^T x^{(k-1)}), && b^{(k)} \in \mathbb{R}^{d_{h}}, &&w^{(k)} \in \mathbb{R}^{d_{h} \times d_{h}} \\
    & x^{(K)} = \phi(b^{(K)} + (w^{(K)})^T x^{(K-1)}),  && b^{(K)} \in \mathbb{R}, &&w^{(K)} \in \mathbb{R}^{d_{h}} \\
    & \hat{t}(x) = \sigma(x^{(K)})
\end{align*}
for $2\leq k \leq K-1$. For clarification, the activation function $\phi(\cdot)$ is applied entry-wise. Furthermore, the tanh function is selected as the activation function $\phi(x) = \frac{e^x - e^{-x}}{e^x + e^{-x}}$, and the Sigmoid function is applied last, since we scale the data set to [0,1] before training. Here, $d, d_{h}$ are, respectively, the dimensions of the covariate and the hidden node.
First, unlike \citet{xia2020neuralfdr}, we do not divide the data set into train, validation, and test data sets, but rather use the whole set for training and testing. This is available because the p-values we derived are based on the empirical distribution and the assumption of \eqref{eqn:symmetric} in  $\mathbf{A1}$. Next, instead of adding a penalty term of $V_\sigma(t) - \alpha R_\sigma(t)$, we use an augmented Lagrangian optimization scheme introduced in \citet{nocedal2006numerical}, for faster convergence.
Specifically, we take the loss function
\begin{equation}
    L(\hat{t}, \lambda_2) = -R_\sigma(\hat{t}) + \lambda_2 (V_\sigma(\hat{t}) - \alpha R_\sigma(\hat{t})) + \frac{\rho}{2}(V_\sigma(\hat{t})- \alpha R_\sigma(\hat{t}))^2
\end{equation}
and update $\lambda_2$ every epoch, using the full data, as follows:
\begin{equation}
    \lambda_2 \leftarrow \max(0, \lambda_2 + \eta(V_\sigma(\hat{t})- \alpha R_\sigma(\hat{t}))).
\end{equation}

\begin{remark}
Applying base knowledge to the network model is possible for the network model. For example, if periodicity with respect to the covariate is expected, we may add a layer of Fourier-featuring(\citet{tancik2020fourier}) to the multilayer perceptron model. However, in the simulation settings in Section \ref{sec:simulation}, adding a layer of Fourier features did not enhance performance, implying that the layer is redundant.    
\end{remark}

For the simulations in  the following section \ref{sec:simulation}, since the covariates are in $\mathbb{R}^1$, we simply use the two layer MLP, for it turned out that they do not show significant difference with a deeper MLP considering the heavier computation time.


\begin{algorithm}[H]
    \caption{Estimating the threshold function $t(x)$}\label{alg:alg2}
    \KwIn{$\{(X_j, p_j, t_0(X_j))\}_{j=1}^n$}
    \KwOut{Rejection set $\mathcal{R}$}
    
    Pretrain $\hat{t}(x)$ to approximate $\{(X_j, t_0(X_j))\}_{j=1}^n$.
    
    Optimize $\hat{t}(x)$ by solving: $\mathbf{maximize}\; \hat{D}(t) \;\; \mathbf{s.t.}\; \widehat{FDP}(t) \leq \alpha$.
    
    Construct the rejection set: $\mathcal{R} \leftarrow \{j \mid Y_j > t(X_j)\}$.
    
    \KwRet{$\mathcal{R}$}
\end{algorithm}

\section{Simulations}
\label{sec:simulation}

For this section, we compare our method with other methods for controlling FDR(AdaPTglm, Adafdr, IHW, SIM). Since the existing methods are applicable to the p-values, the {\bf Algorithm} ~\ref{alg:alg1} is applied first, then we compare the performance of such methods and {\bf Algorithm} ~\ref{alg:alg2}.


The number of samples derived from null and alternative distributions is, respectively, 4000 and 1000. The covariate $X_i$ for the null distribution is derived from the uniform distribution, while for the alternative distributions, it varies in the simulation settings. For the observed values($y_i$), we first sample the quantile $q_i$ using the beta distribution. For the null values, since the distribution should be symmetric, the two parameters for the beta distribution are the same, while for the alternative values, we choose the first parameter to be significantly larger than the other. Then, we take some distribution $\Phi_i$ that may or may not depend on the covariate$X_i$ and derive the observed values. Each simulation is done 500 times.

\begin{table}[htbp]
\centering
\caption{Progressive expansion of covariate influence across simulation settings. 
A checkmark (\checkmark) indicates direct dependence of that component on the covariate $X_i$.}
\label{tab:covariate_expansion}
\begin{tabular}{@{}ccccc@{}}
\toprule
\textbf{Setting}                       & \textbf{$f
_{1,x}$} & \textbf{$f_{0,x}$} & \textbf{Signal proportion $\pi(x)$} & \textbf{Periodic effect} \\ \midrule
\textbf{Section~\ref{sec:simulation1}} & \checkmark                        & $\cdot$                     & $\cdot$                              & $\cdot$                               \\
\textbf{Section~\ref{sec:simulation2}} & \checkmark                        & \checkmark                 & $\cdot$                              & $\cdot$                               \\
\textbf{Section~\ref{sec:simulation3}} & \checkmark                        & \checkmark                 & \checkmark                          & $\cdot$                               \\
\textbf{Section~\ref{sec:simulation4}} & \checkmark                        & \checkmark                 & \checkmark                          & \checkmark                           \\ \bottomrule
\end{tabular}
\end{table}

{
Building on this common framework, we progressively increase the complexity of how the covariate affects the data-generating process. 
The first setting begins with a minimal dependence where the covariate influences only the alternative distribution. 
We then extend this influence to both the null and alternative distributions, further introduce heterogeneity in the signal proportion across the covariate space, 
and finally impose nonlinear, periodic dependence patterns. 
This gradual escalation allows a systematic evaluation of whether each method can sustain FDR control and detection power under increasingly complex covariate structures.
}

\subsection{Covariate affecting only the alternative distribution}\label{sec:simulation1}
In the baseline setting, the covariate influences only the \textit{alternative} distribution, while the null distribution remains fixed. 
Each covariate $X_i$ is drawn uniformly from $[0,1]$. 
For null hypotheses, the observed values $y_i$ are generated from symmetric $\mathrm{Beta}(2,2)$ quantiles transformed through a truncated Normal distribution $\Phi_i = N(10, 10)_{[-2.5,\, 2.5]}$, producing stationary nulls that do not vary with the covariate. 
For alternative hypotheses, right-skewed $\mathrm{Beta}(10,0.5)$ quantiles are mapped through $\Phi_i = N(5e^{X_i},\, 10 - \sin(\pi X_i))_{[-2.5,\, 2.5]}$, in which the mean increases exponentially with $X_i$ and the variance oscillates mildly with $X_i$. 
This isolates cases where only the signal distribution varies across the covariate, testing whether methods can adapt to signal-specific dependence without covariate-driven shifts in the null.

\begin{table}[H]
\centering
\resizebox{\textwidth}{!}{%
\begin{tabular}{cccccccccc}
\hline
 & \multicolumn{3}{c}{$\alpha = 0.05$} & \multicolumn{3}{c}{$\alpha = 0.10$} & \multicolumn{3}{c}{$\alpha = 0.20$} \\ \cline{2-10}
Method & $R$ & $\widehat{FDR}$ & $TPR$ & $R$ & $\widehat{FDR}$ & $TPR$ & $R$ & $\widehat{FDR}$ & $TPR$ \\ \hline
Proposed Method & \textbf{326 (57.6)} & \textbf{0.0369 (0.0149)} & \textbf{0.314 (0.0531)} &
\textbf{443 (70.2)} & \textbf{0.0541 (0.0176)} & \textbf{0.418 (0.0626)} &
\textbf{633 (73.3)} & \textbf{0.0918 (0.0228)} & \textbf{0.573 (0.0568)} \\
SIM              & 206 (49.4) & \textbf{0.0409 (0.0211)} & 0.198 (0.0466) &
242 (60.3) & \textbf{0.0616 (0.0258)} & 0.226 (0.0549) &
342 (88.3) & \textbf{0.113 (0.0309)} & 0.302 (0.0733) \\
AdaFDR           & 287 (56.5) & \textbf{0.0373 (0.0165)} & 0.276 (0.0522) &
378 (70.6) & \textbf{0.0595 (0.0220)} & 0.355 (0.0628) &
538 (87.7) & \textbf{0.110 (0.0295)} & 0.478 (0.0730) \\
AdaPT GLM        & 280 (54.8) & \textbf{0.0318 (0.0146)} & 0.270 (0.0509) &
381 (69.8) & \textbf{0.0521 (0.0188)} & 0.361 (0.0625) &
560 (90.4) & \textbf{0.0960 (0.0236)} & 0.505 (0.0755) \\
IHW              & 245 (57.8) & \textbf{0.0257 (0.0125)} & 0.239 (0.0548) &
361 (74.5) & \textbf{0.0385 (0.0147)} & 0.347 (0.0685) &
583 (72.5) & \textbf{0.0747 (0.0199)} & 0.539 (0.0582) \\ \hline
\end{tabular}%
}
\caption{Comparison of methods across different $\alpha$ values. Reported as mean (std). Bold indicates the highest R, highest TPR, and all $\widehat{FDR}$ values below the nominal $\alpha$.}
\label{tab:result_custom1}
\end{table}

\subsection{Covariate affecting both null and alternative distributions}\label{sec:simulation2}
In this setting, both null and alternative distributions vary with the covariate, introducing heterogeneity in the overall data scale. 
We again sample $X_i$ uniformly from $[0,1]$. 
Null observations are derived from $\mathrm{Beta}(2,2)$ quantiles transformed through $\Phi_i = N(10e^{X_i},\, 5 + \sin(\pi X_i))_{[-2.5,\, 2.5]}$, so that both the mean and variance depend smoothly on $X_i$. 
For alternatives, right-skewed $\mathrm{Beta}(10,0.5)$ quantiles are transformed through the same family $\Phi_i$, ensuring that both null and alternative distributions co-vary with the covariate. 
Because the entire data distribution shifts as a function of $X_i$, this scenario challenges methods to maintain calibration when the null itself is covariate-dependent.

\begin{table}[H]
\centering
\resizebox{\textwidth}{!}{%
\begin{tabular}{cccccccccc}
\hline
 & \multicolumn{3}{c}{$\alpha = 0.05$} & \multicolumn{3}{c}{$\alpha = 0.10$} & \multicolumn{3}{c}{$\alpha = 0.20$} \\ \cline{2-10}
Method & $R$ & $\widehat{FDR}$ & $TPR$ & $R$ & $\widehat{FDR}$ & $TPR$ & $R$ & $\widehat{FDR}$ & $TPR$ \\ \hline
Proposed Method & \textbf{625 (51.8)} & \textbf{0.0356 (0.0112)} & \textbf{0.603 (0.0454)} &
\textbf{752 (50.4)} & \textbf{0.0621 (0.0153)} & \textbf{0.705 (0.0386)} &
\textbf{926 (51.4)} & \textbf{0.118 (0.0212)} & \textbf{0.816 (0.0303)} \\
SIM              & 442 (61.5) & \textbf{0.0438 (0.0164)} & 0.423 (0.0561) &
519 (66.9) & \textbf{0.0793 (0.0209)} & 0.477 (0.0573) &
677 (77.7) & \textbf{0.152 (0.0280)} & 0.573 (0.0561) \\
AdaFDR           & 594 (56.0) & \textbf{0.0322 (0.0111)} & 0.574 (0.0500) &
728 (52.5) & \textbf{0.0572 (0.0144)} & 0.686 (0.0420) &
911 (53.6) & \textbf{0.113 (0.0209)} & 0.808 (0.0324) \\
AdaPT GLM        & 591 (55.8) & \textbf{0.0296 (0.0102)} & 0.573 (0.0502) &
731 (52.0) & \textbf{0.0545 (0.0143)} & 0.691 (0.0413) &
916 (52.5) & \textbf{0.109 (0.0202)} & 0.815 (0.0314) \\
IHW              & 539 (56.0) & \textbf{0.0234 (0.0089)} & 0.527 (0.0520) &
681 (51.1) & \textbf{0.0442 (0.0124)} & 0.650 (0.0428) &
864 (49.0) & \textbf{0.0922 (0.0176)} & 0.783 (0.0326) \\ \hline
\end{tabular}%
}
\caption{Comparison of methods across different $\alpha$ values. Reported as mean (std). Bold indicates the highest R, highest TPR, and all $\widehat{FDR}$ values below the nominal $\alpha$.}
\label{tab:result_custom2}
\end{table}

\subsection{Covariate affecting the mixture proportion of signals}\label{sec:simulation3}
Here, the covariate determines the local density of signals, creating heterogeneous signal proportions. 
The covariates for nulls are drawn uniformly, whereas those for alternatives follow a $\mathrm{Beta}(2,2)$ distribution that concentrates around the midpoint of the covariate range. 
Null $y_i$ values are obtained from $\mathrm{Beta}(3,3)$ quantiles transformed through $\Phi_i = N(10e^{X_i},\, 5 + \sin(\pi X_i))_{[-2.5,\, 2.5]}$, while alternative values use $\mathrm{Beta}(10,0.5)$ quantiles through the same covariate-dependent $\Phi_i$. 
Thus, both distributions evolve with $X_i$, but the proportion of signals is higher around $X_i = 0.5$. 
This design evaluates whether methods can adaptively learn thresholds that exploit uneven signal concentration while still maintaining FDR control.

\begin{table}[H]
\centering
\resizebox{\textwidth}{!}{%
\begin{tabular}{cccccccccc}
\hline
 & \multicolumn{3}{c}{$\alpha = 0.05$} & \multicolumn{3}{c}{$\alpha = 0.10$} & \multicolumn{3}{c}{$\alpha = 0.20$} \\ \cline{2-10}
Method & $R$ & $\widehat{FDR}$ & $TPR$ & $R$ & $\widehat{FDR}$ & $TPR$ & $R$ & $\widehat{FDR}$ & $TPR$ \\ \hline
Proposed Method & \textbf{871 (32.1)} & \textbf{0.0438 (0.0111)} & \textbf{0.832 (0.0244)} &
\textbf{957 (31.6)} & \textbf{0.0773 (0.0157)} & \textbf{0.883 (0.0185)} &
\textbf{1.09e+03 (37.5)} & \textbf{0.148 (0.0221)} & \textbf{0.928 (0.0149)} \\
SIM              & 709 (47.7) & {0.0633 (0.0143)} & 0.664 (0.0415) &
804 (49.7) & {0.109 (0.0185)} & 0.715 (0.0379) &
1.01e+03 (56.5) & \textbf{0.198 (0.0218)} & 0.807 (0.0323) \\
AdaFDR           & 854 (34.0) & \textbf{0.0374 (0.0102)} & 0.822 (0.0267) &
947 (31.4) & \textbf{0.0684 (0.0141)} & 0.881 (0.0192) &
1.08e+03 (36.5) & \textbf{0.134 (0.0204)} & 0.935 (0.0128) \\
AdaPT GLM        & 847 (33.9) & \textbf{0.0377 (0.0101)} & 0.815 (0.0271) &
942 (32.4) & \textbf{0.0695 (0.0144)} & 0.876 (0.0198) &
1.08e+03 (37.8) & \textbf{0.137 (0.0209)} & 0.932 (0.0134) \\
IHW              & 809 (33.4) & \textbf{0.0282 (0.0083)} & 0.786 (0.0283) &
903 (30.8) & \textbf{0.0544 (0.0123)} & 0.853 (0.0216) &
1.03e+03 (33.9) & \textbf{0.112 (0.0188)} & 0.913 (0.0150) \\ \hline
\end{tabular}%
}
\caption{Comparison of methods across different $\alpha$ values. Reported as mean (std). Bold indicates the highest R, highest TPR, and all $\widehat{FDR}$ values below the nominal $\alpha$.}
\label{tab:result_custom3}
\end{table}

\subsection{Nonlinear and periodic covariate dependence}\label{sec:simulation4}
Finally, the most complex scenario introduces nonlinear and periodic effects of the covariate on both distributions and signal proportion. 
Null covariates are sampled uniformly on $[0,1]$, while alternative covariates follow a $\mathrm{Beta}(0.5,0.5)$ distribution, yielding higher signal density near the boundaries. 
For nulls, symmetric $\mathrm{Beta}(2,2)$ quantiles are transformed through $\Phi_i = N(\sin(4X_i) + \sin(8X_i),\, 5)_{[-2.5,\, 2.5]}$, and for alternatives, right-skewed $\mathrm{Beta}(10,0.5)$ quantiles are passed through the same periodic $\Phi_i$. 
Because the mean of $\Phi_i$ oscillates multiple times within the covariate range, this scenario tests whether methods can track highly nonlinear, non-monotonic dependencies in both signal distribution and signal prevalence.


\begin{table}[H]
\centering
\resizebox{\textwidth}{!}{%
\begin{tabular}{cccccccccc}
\hline
 & \multicolumn{3}{c}{$\alpha = 0.05$} & \multicolumn{3}{c}{$\alpha = 0.10$} & \multicolumn{3}{c}{$\alpha = 0.20$} \\ \cline{2-10}
Method & $R$ & $\widehat{FDR}$ & $TPR$ & $R$ & $\widehat{FDR}$ & $TPR$ & $R$ & $\widehat{FDR}$ & $TPR$ \\ \hline
Proposed Method & \textbf{524 (68.7)} & \textbf{0.0373 (0.0128)} & \textbf{0.504 (0.0636)} &
\textbf{643 (72.7)} & \textbf{0.0623 (0.0163)} & \textbf{0.603 (0.0637)} &
\textbf{817 (76.8)} & \textbf{0.116 (0.0241)} & \textbf{0.721 (0.0583)} \\
SIM              & 415 (65.0) & \textbf{0.0202 (0.0106)} & 0.407 (0.0615) &
483 (68.3) & \textbf{0.0371 (0.0168)} & 0.464 (0.0611) &
599 (78.5) & \textbf{0.0896 (0.0287)} & 0.544 (0.0601) \\
AdaFDR           & 494 (68.3) & \textbf{0.0356 (0.0131)} & 0.476 (0.0635) &
608 (73.8) & \textbf{0.0602 (0.0162)} & 0.571 (0.0659) &
780 (85.4) & \textbf{0.115 (0.0233)} & 0.689 (0.0688) \\
AdaPT GLM        & 483 (69.1) & \textbf{0.0327 (0.0126)} & 0.467 (0.0641) &
607 (72.9) & \textbf{0.0572 (0.0160)} & 0.571 (0.0640) &
789 (80.9) & \textbf{0.112 (0.0224)} & 0.700 (0.0622) \\
IHW              & 451 (69.3) & \textbf{0.0269 (0.0111)} & 0.438 (0.0652) &
574 (70.4) & \textbf{0.0479 (0.0143)} & 0.546 (0.0633) &
753 (77.1) & \textbf{0.0967 (0.0200)} & 0.679 (0.0619) \\ \hline
\end{tabular}%
}
\caption{Comparison of methods across different $\alpha$ values. Reported as mean (std). Bold indicates the highest R, highest TPR, and all $\widehat{FDR}$ values below the nominal $\alpha$.}
\label{tab:result_custom4}
\end{table}

\section{Real Data}
\label{sec:realdata}
\subsection{ Application to Blood Pressure Data}
\label{subsec:SBP}

This study utilized publicly available data from the National Health and Nutrition Examination Survey (NHANES) 2021--2023 cohort~\citep{AgeBP}. Our analysis focuses on the relationship between systolic blood pressure (SBP) and age. A key assumption of our proposed methodology is that the conditional null distribution of SBP, given a covariate like age, is symmetric.

To empirically validate this assumption, we first visualized the overall relationship between age and SBP using a scatter plot (Figure~\ref{fig:bp_a}). By plotting the SBP density across different age categories, we created the raincloud plot shown in Figure~\ref{fig:bp_b}. As the figure demonstrates, the SBP distributions for each age group exhibit approximate symmetry, albeit around a center that shifts with age. This visual evidence provides strong justification for the application of our symmetry-based method.

\begin{figure}[h!]
    \centering

    \begin{subfigure}[b]{0.48\textwidth}
        \centering
        \includegraphics[width=\textwidth]{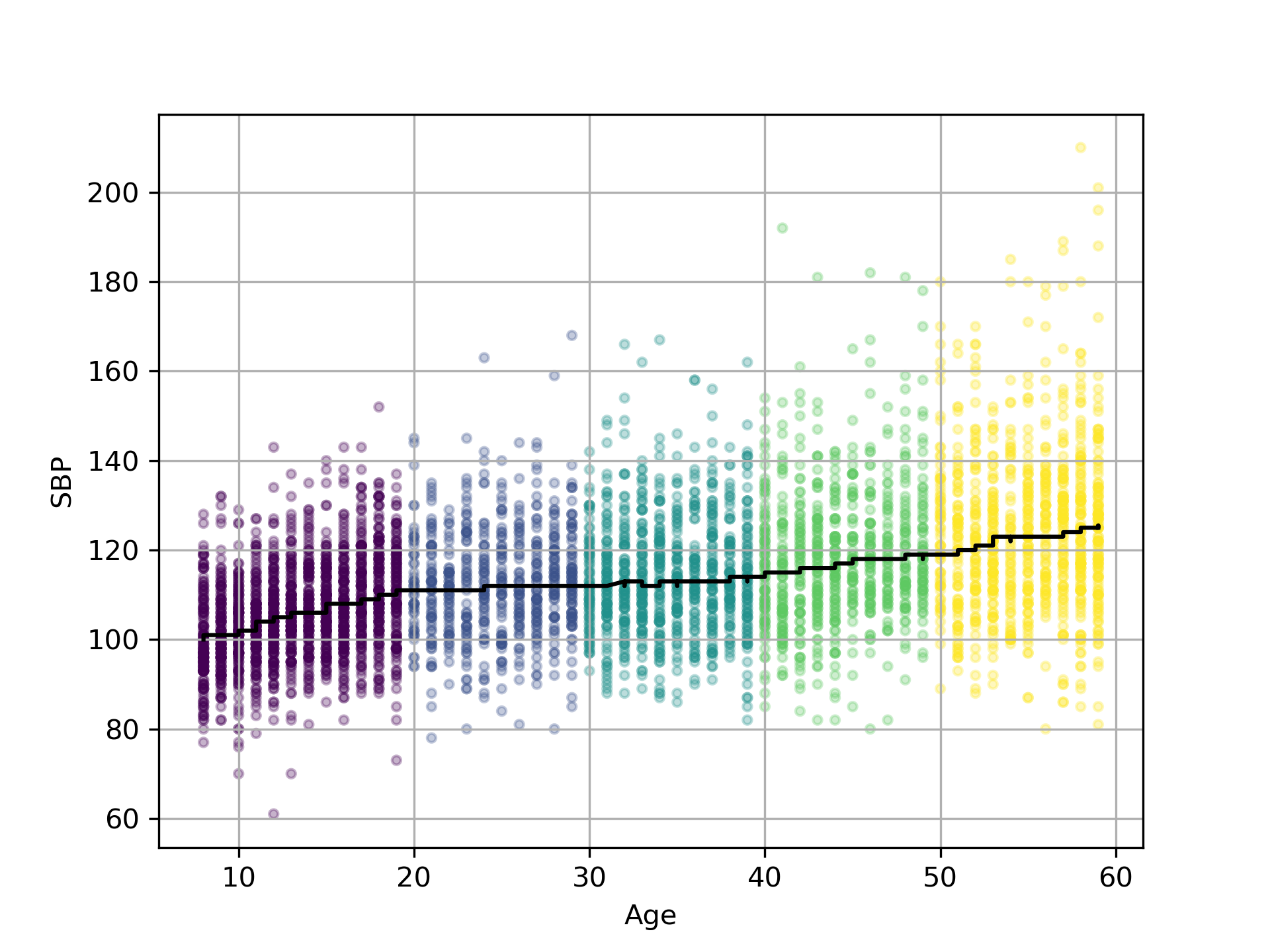}
        \caption{Scatter plot of age versus SBP.}
        \label{fig:bp_a}
    \end{subfigure}
    \hfill
    \begin{subfigure}[b]{0.48\textwidth}
        \centering
        \includegraphics[width=\textwidth]{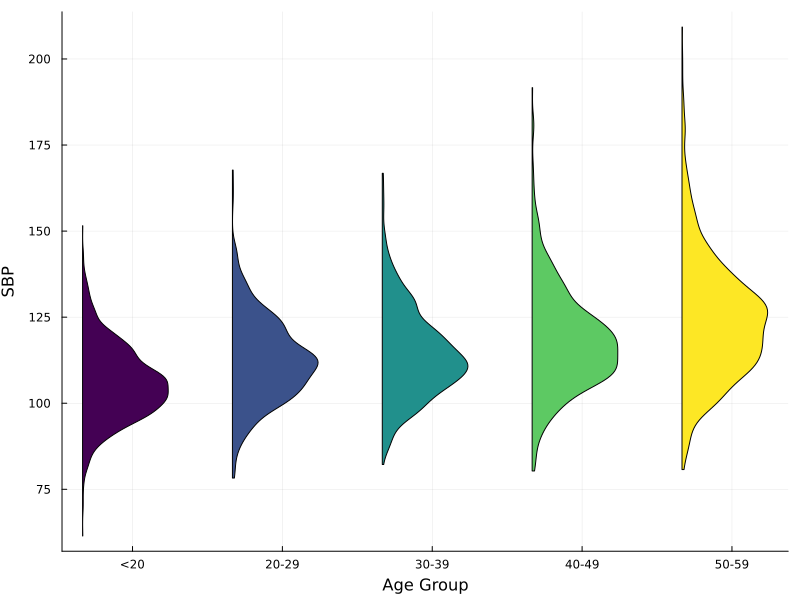}
        \caption{Raincloud plot of SBP distributions.}
        \label{fig:bp_b}
    \end{subfigure}
    
    \caption[Validation of Symmetry Assumption]{
        Validation of the conditional symmetry assumption using NHANES SBP data. 
    }
    \label{fig:bp_validation} 
\end{figure}

\begin{figure}[htbp]
    \centering
    \begin{minipage}[b]{0.45\linewidth}
        \centering
        \includegraphics[width=\linewidth]{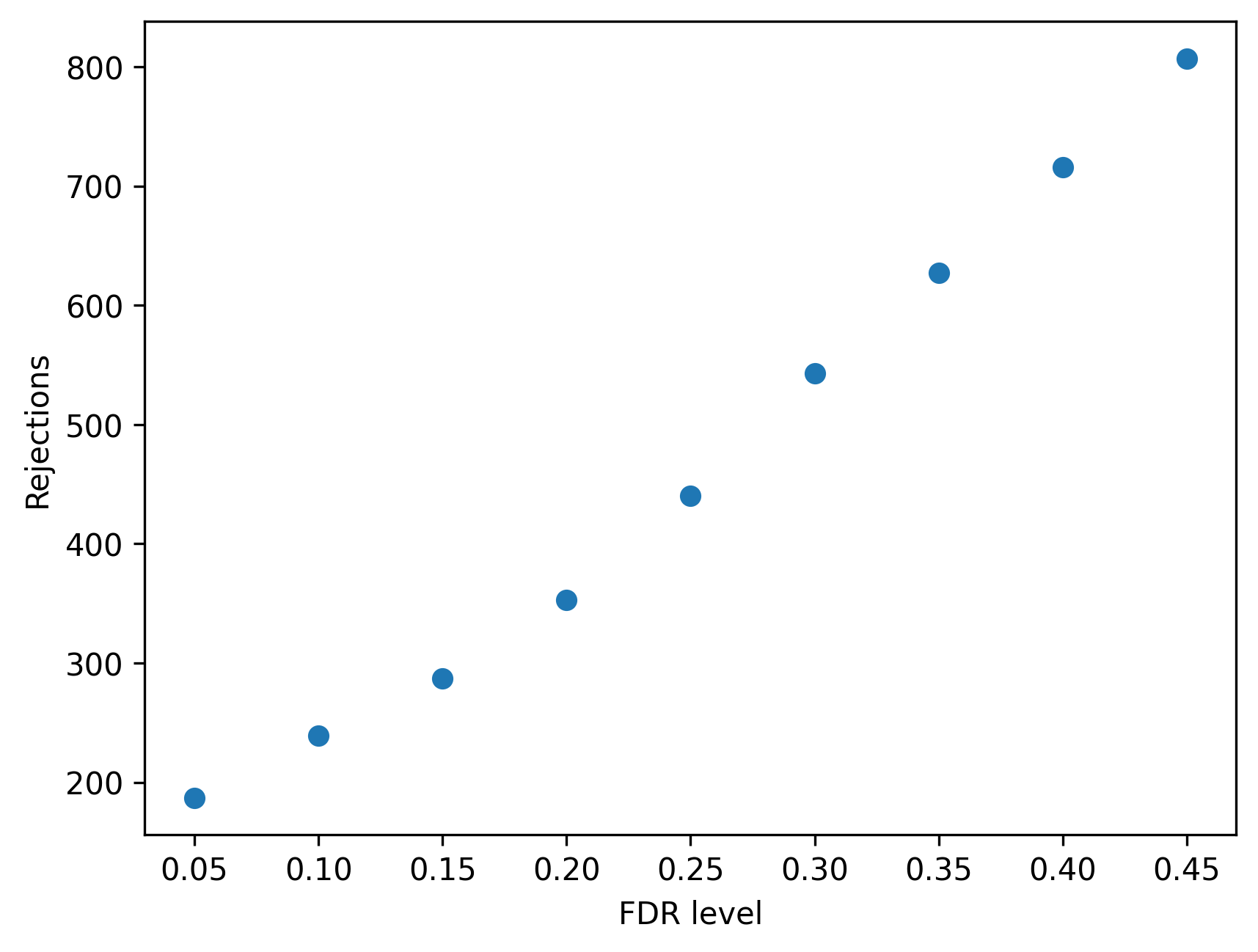}
        \caption{Rejections for each FDR level}
        \label{fig:SBP}
    \end{minipage}
    \hfill
    \begin{minipage}[b]{0.45\linewidth}
        \centering
        \includegraphics[width=\linewidth]{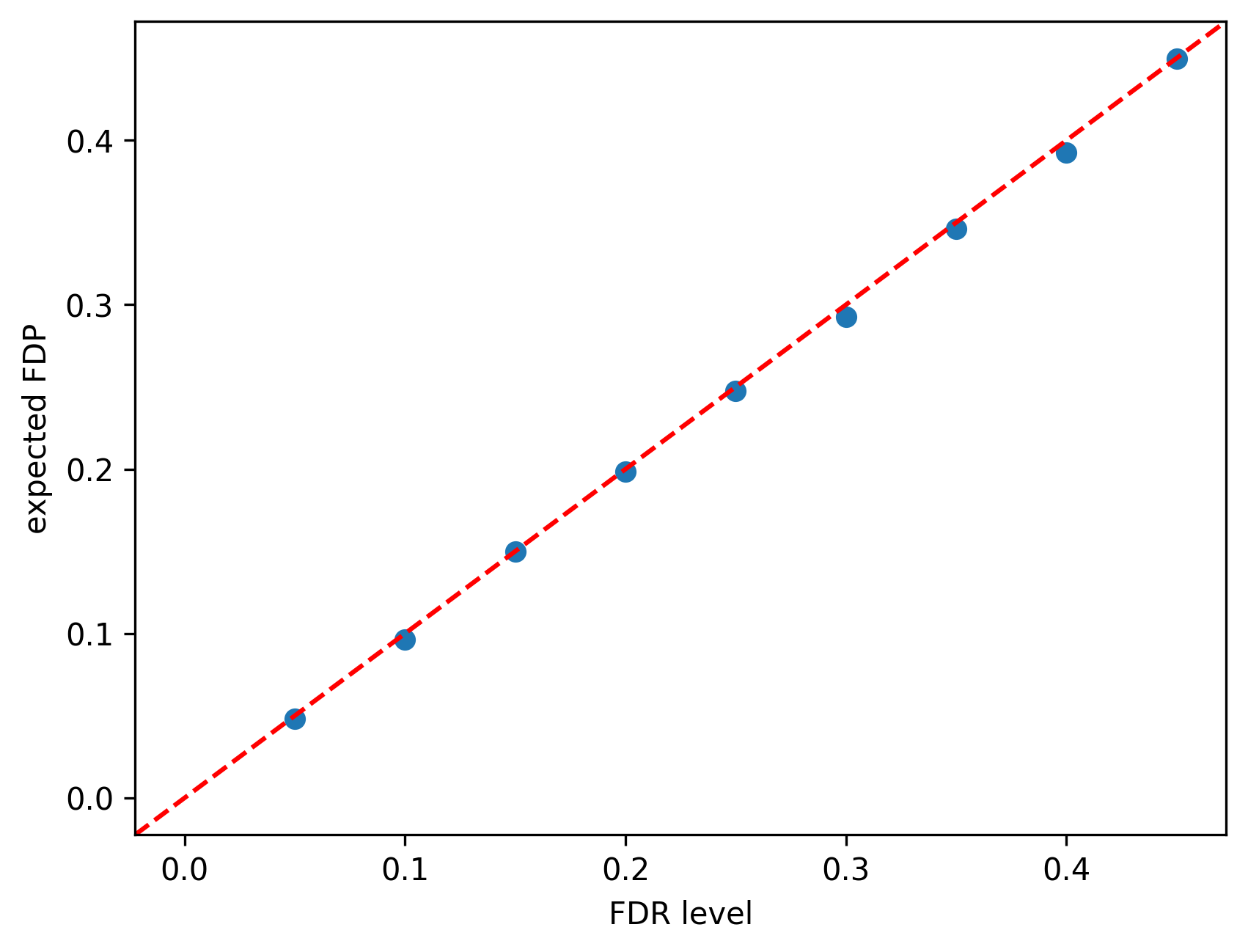}
        \caption{$\hat{\text{FDP}}$ for each FDR level}
        \label{fig:SBP2}
    \end{minipage}
\end{figure}

{

\begin{table}[htbp]
\caption{Number of Hypertensive Patients by Age Group and Diagnostic Method (Age $< 60$)}
\label{tab:patient_counts}
\begin{tabular}{@{}cccccccc@{}}
\toprule
\multicolumn{2}{c}{Age Group}                             & $<$20 & 20-29 & 30-39 & 40-49 & 50-59 & Total \\ \midrule
\multicolumn{2}{c}{Conventional Criterion (SBP $\ge$140mmHg)} & 5   & 12    & 30    & 60    & 142   & 249   \\
\multirow{3}{*}{Proposed Method}     & $\alpha = 0.05$    & 25  & 6     & 32    & 25    & 39    & 127   \\
                                     & $\alpha = 0.10$    & 43  & 14    & 41    & 29    & 39    & 166   \\
                                     & $\alpha = 0.20$    & 71  & 33    & 51    & 29    & 40    & 229   \\ \bottomrule
\end{tabular}
\end{table}

Before comparing our results with those obtained under the conventional threshold, it is important to emphasize that the traditional approach does not account for false discovery rate (FDR) control. Therefore, the outcomes of our method should be interpreted within its explicitly FDR-adjusted framework.

Across all significance levels, the proposed diagnostic procedure exhibited distinct rejection patterns compared with the conventional SBP $\ge 140$ mmHg rule. In the age-stratified analysis, our method consistently identified a higher proportion of hypertensive cases among younger individuals ($<$20 and 20–29) than the conventional criterion. Notably, in the youngest group ($<$20), the number of detections increased sharply from 25 at $\alpha=0.05$ to 43 at $\alpha=0.10$ and 71 at $\alpha=0.20$, suggesting that many individuals in this cohort have borderline SBP values close to the rejection threshold.

The increase in diagnoses among younger participants and the corresponding decrease among older cohorts align with previous findings that challenge the validity of a single SBP cutoff across all ages. Elevated SBP levels in youth—even those below 140 mmHg—have been associated with long-term cardiovascular risk \citep{yano2018hypertension, pletcher2008prehypertension, chen2008tracking}, whereas overly aggressive blood pressure reduction in older adults may lead to adverse effects \citep{min2015aggressive, mallery2025promoting, arguedas2013aggressive}. These results highlight the importance of age-adaptive diagnostic criteria, which our method seeks to formalize for individuals under 60 years of age.


\subsection{Application to US Air Pollution Data}

To demonstrate the practical utility and flexibility of our proposed method, we applied it to daily PM2.5 measurements from the U.S. Environmental Protection Agency (EPA) AirData database, which compiles nationwide observations from the Air Quality System (AQS) network \citep{EPA_AirData}. PM2.5, or fine particulate matter with diameters of 2.5~$\mu$m or smaller, poses serious risks to respiratory and cardiovascular health \citep{xing2016impact}. Our goal was to identify anomalous pollution events under two different covariate settings: (\textit{i}) time (month) and (\textit{ii}) spatial location (longitude and latitude).

\paragraph{(i) Temporal covariate (month).}
The first analysis used time as a one-dimensional covariate to model the gradual seasonal variation in the null center $m(t)$. 
Raw PM2.5 data exhibit significant right-skewness, violating the conditional symmetry assumption. 
To mitigate this, we applied a log-transformation to the concentration values, which successfully symmetrized the distribution, as illustrated in Figure~\ref{fig:pm25_distribution}. 
This transformation enables direct application of our procedure for estimating $m(t)$ and identifying anomalous deviations from it.

\begin{figure}[htbp]
    \centering
    \begin{subfigure}[b]{0.48\textwidth}
        \centering
        \includegraphics[width=\textwidth]{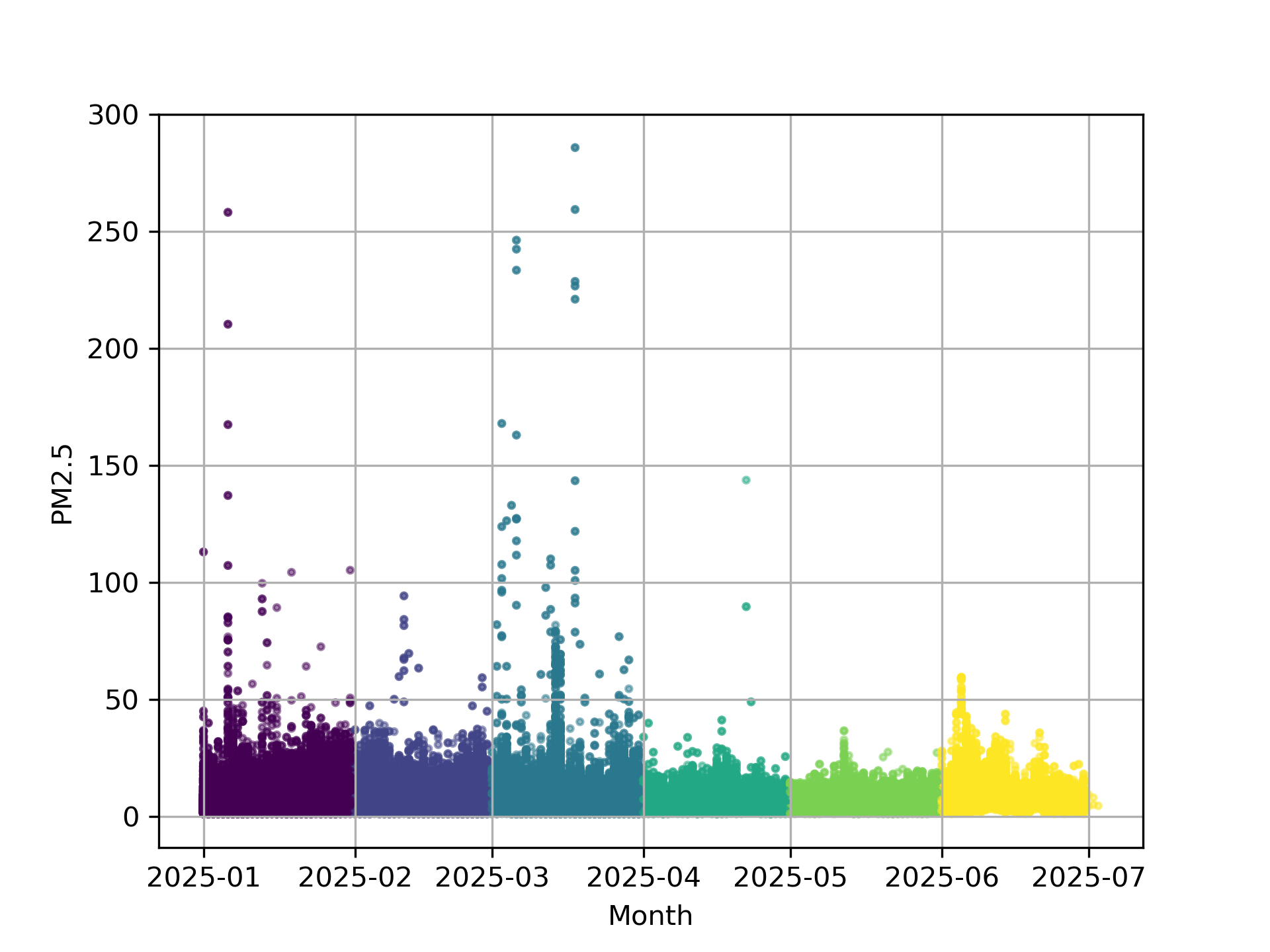}
        \caption*{Original Scatter Plot}
    \end{subfigure}
    \hfill
    \begin{subfigure}[b]{0.48\textwidth}
        \centering
        \includegraphics[width=\textwidth]{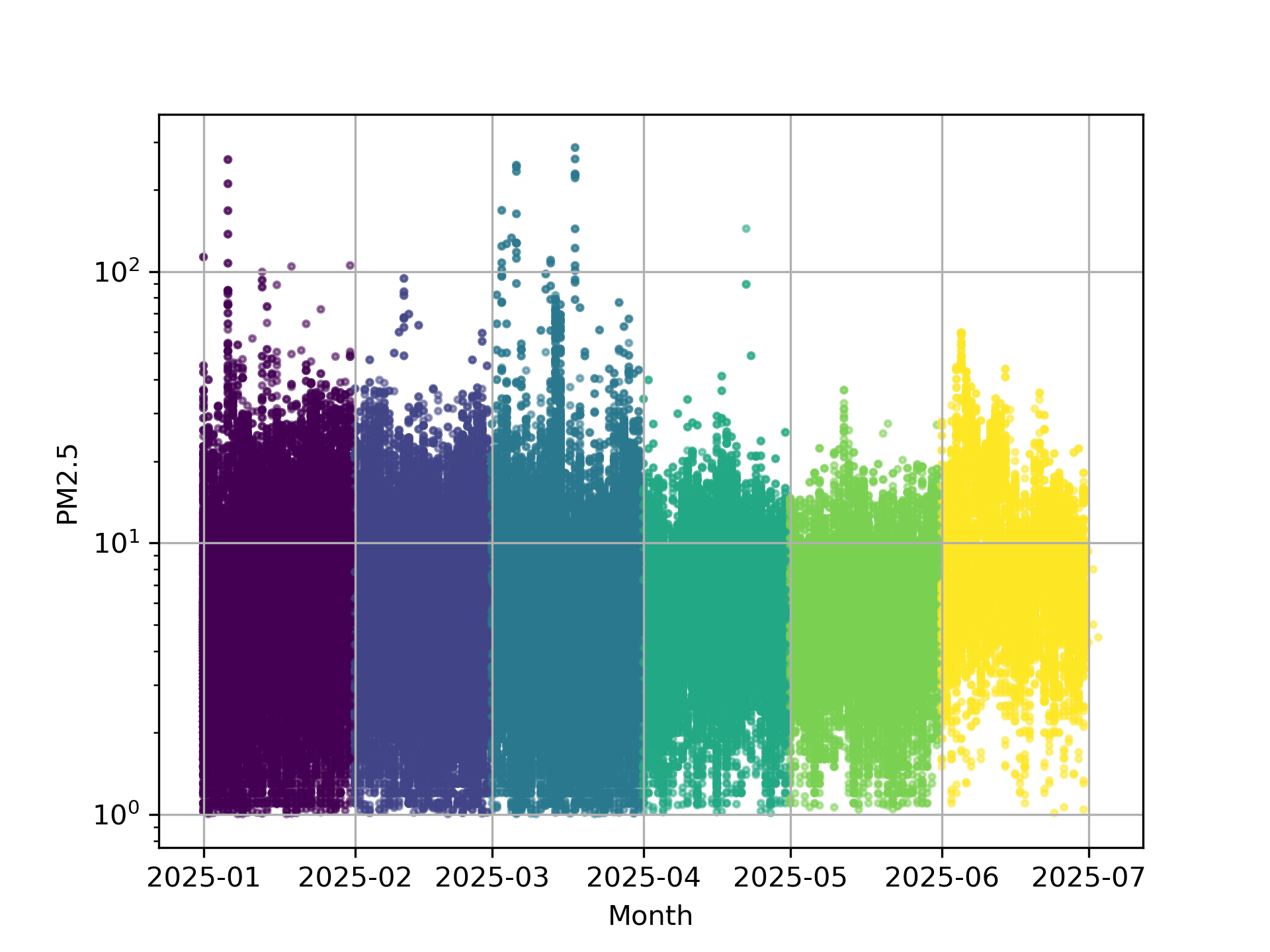}
        \caption*{Log-transformed Scatter Plot}
    \end{subfigure}

    \vspace{0.3cm}

    \begin{subfigure}[b]{0.48\textwidth}
        \centering
        \includegraphics[width=\textwidth]{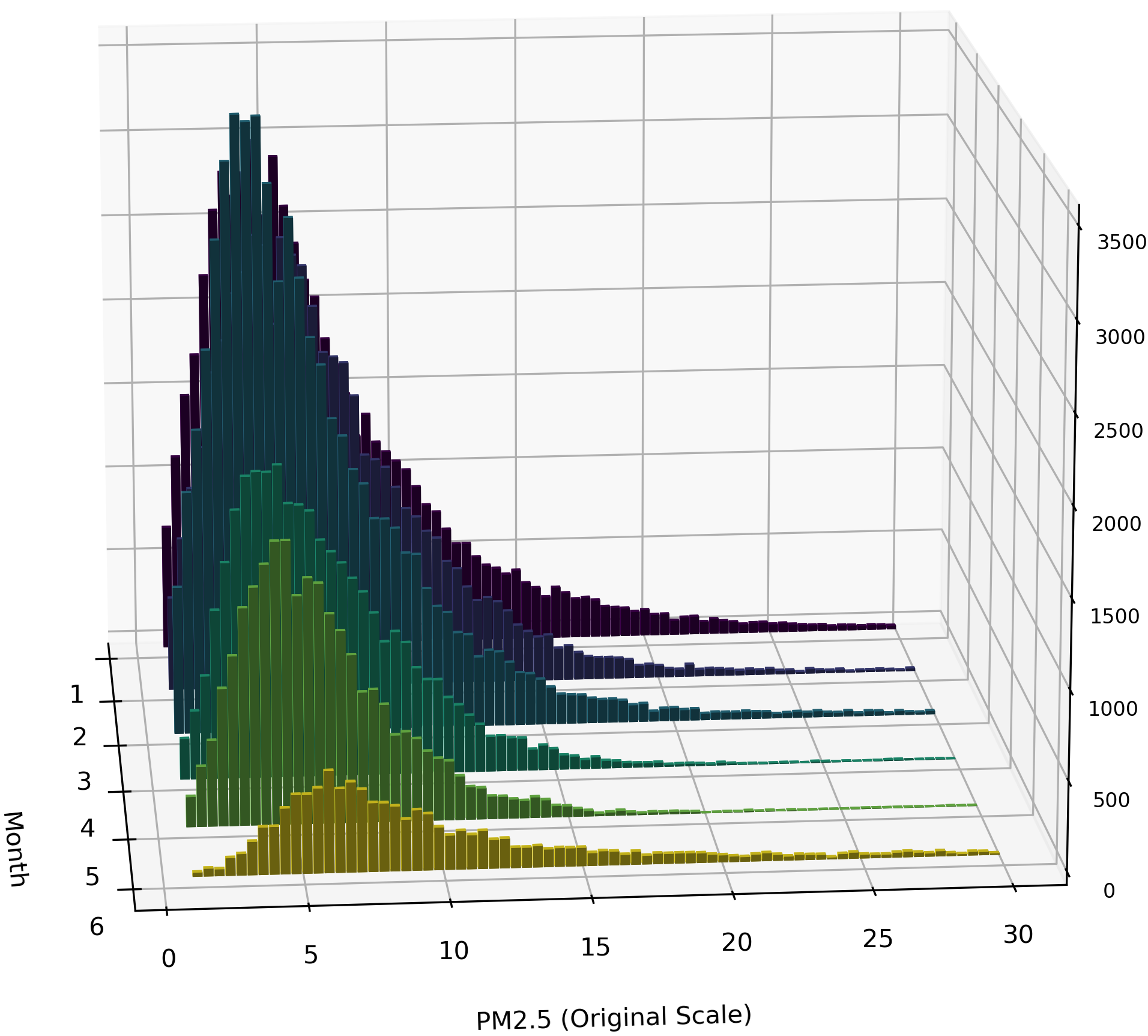}
        \caption*{Original Distribution}
    \end{subfigure}
    \hfill
    \begin{subfigure}[b]{0.48\textwidth}
        \centering
        \includegraphics[width=\textwidth]{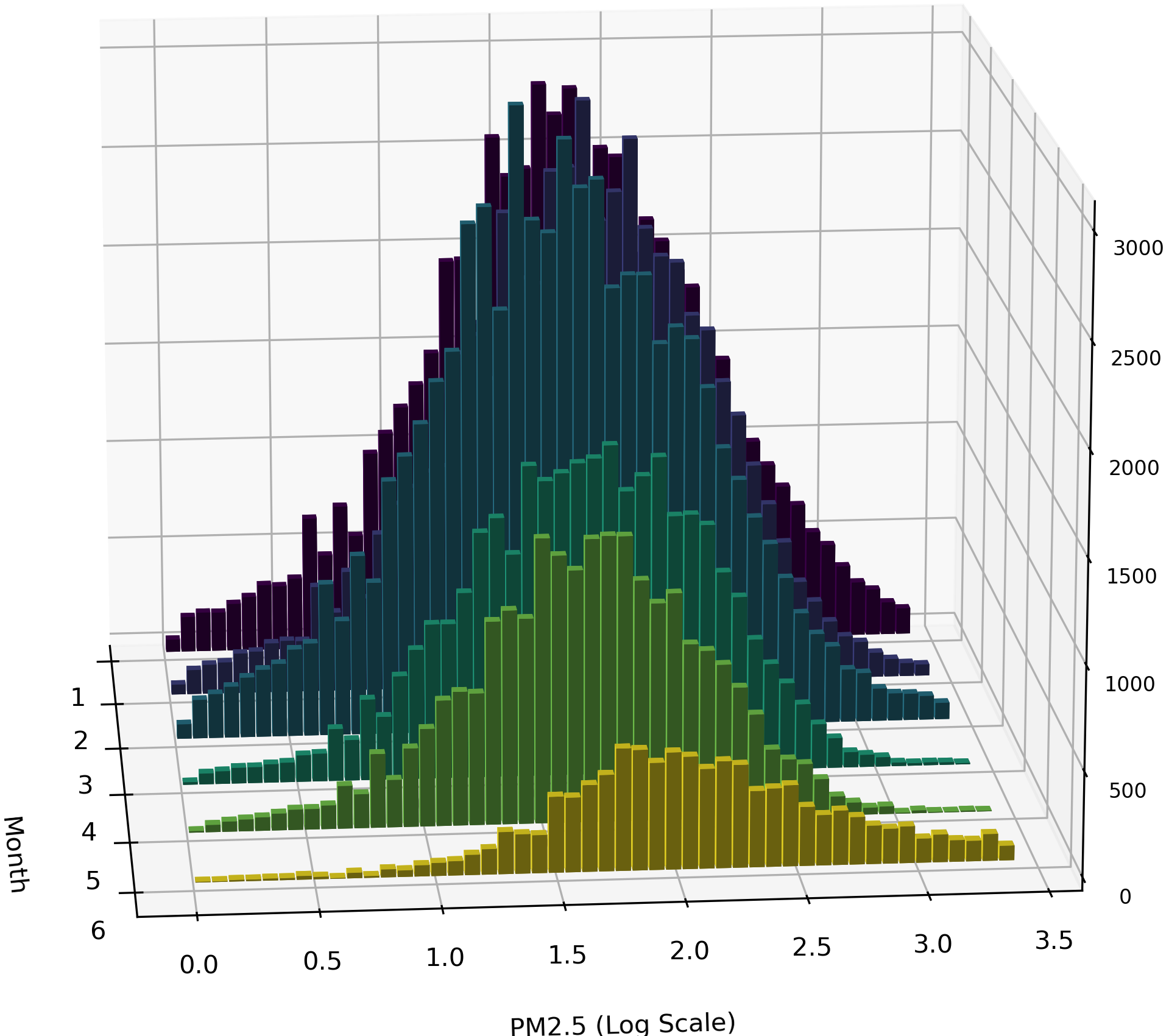}
        \caption*{Log-transformed Distribution}
    \end{subfigure}
    
    \caption{Comparison of original and log-transformed PM2.5 data distributions. 
    The transformation symmetrizes the data, making it suitable for our proposed method.}
    \label{fig:pm25_distribution}
\end{figure}

In this temporal setting, the null hypothesis $H_0$ corresponds to ordinary daily fluctuations in PM2.5, whereas the alternative $H_1$ represents transient spikes potentially caused by events such as wildfires or industrial emissions. 
Figure~\ref{fig:air_pollution_time} shows the monthly rejection results for the first half of 2025, where the color of each point denotes the month of detection. 
The detected anomalies closely align with independently recorded pollution events, including winter wildfires in California, prescribed burns in Florida, and Canadian wildfire smoke in June. 
These findings confirm that our method effectively captures temporally localized events while maintaining FDR control.

\begin{figure}[h!]
    \centering
    \includegraphics[width=\textwidth]{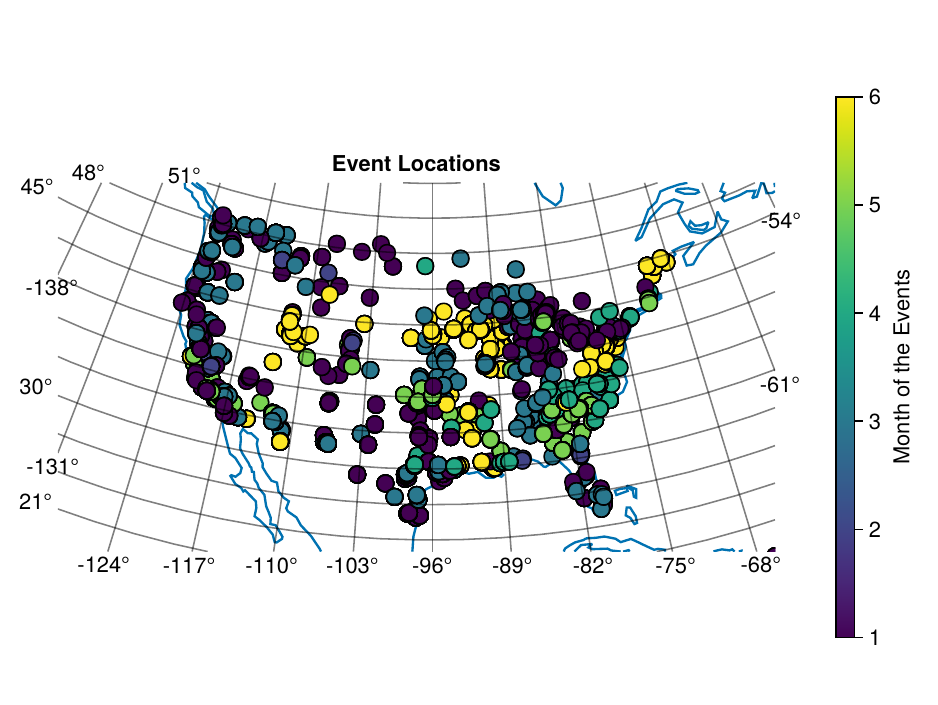}
    \caption{Detected anomalous PM2.5 events across the U.S.\ (January--June 2025), using time as the covariate. 
    Each point is colored by month, revealing seasonal and regional event patterns.}
    \label{fig:air_pollution_time}
\end{figure}

\paragraph{(ii) Spatial covariate (longitude and latitude).}
Next, we extended the analysis by considering two-dimensional spatial covariates, $(\text{longitude},\text{latitude})$. 
For each fixed month, we estimated the conditional null center $m(\text{longitude},\text{latitude})$ and identified locations whose daily PM2.5 concentrations significantly exceeded the estimated local baseline under FDR control.
Figure~\ref{fig:air_pollution_space} visualizes the spatial rejection patterns across six months (January--June). 
Gray dots represent ordinary observations, while colored markers indicate rejections obtained from our procedure. 
The clustering of rejections in the Western and Southeastern regions coincides with independently documented pollution events during the first half of 2025. 
Notably, the three major spatio-temporal patterns identified in the temporal analysis (Figure~\ref{fig:pm25_distribution})---winter wildfires in California, prescribed burns in Florida, and transboundary smoke from Canadian wildfires---are consistently reproduced in this spatial analysis. 
This concordance indicates that our method effectively captures spatially localized and seasonally varying anomalies in air pollution, even under a flexible nonparametric framework.

\begin{figure}[htbp]
    \centering
    \begin{subfigure}[b]{0.32\textwidth}
        \centering
        \includegraphics[width=\textwidth]{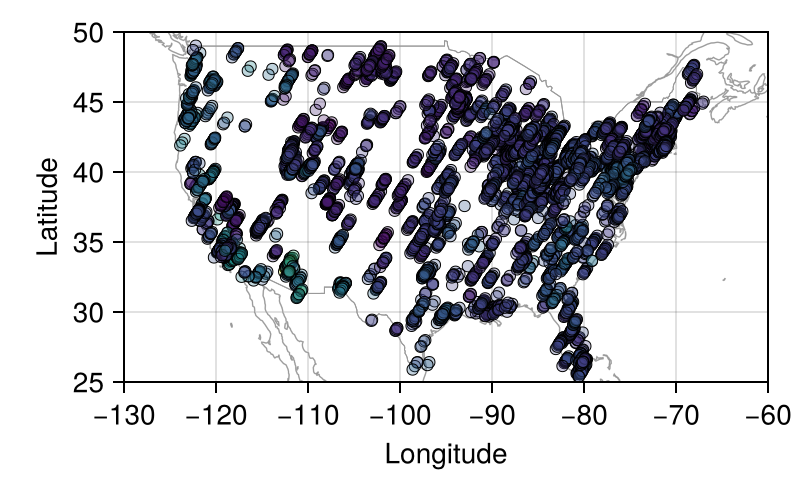}
        \caption*{January}
    \end{subfigure}
    \hfill
    \begin{subfigure}[b]{0.32\textwidth}
        \centering
        \includegraphics[width=\textwidth]{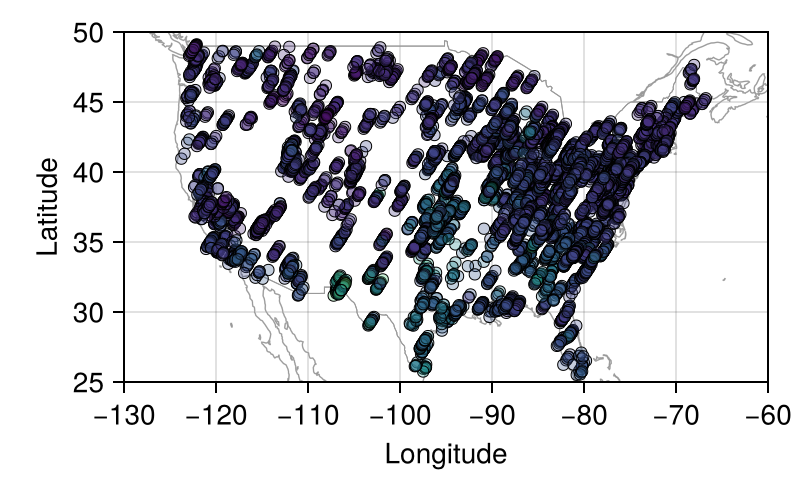}
        \caption*{March}
    \end{subfigure}
    \hfill
    \begin{subfigure}[b]{0.32\textwidth}
        \centering
        \includegraphics[width=\textwidth]{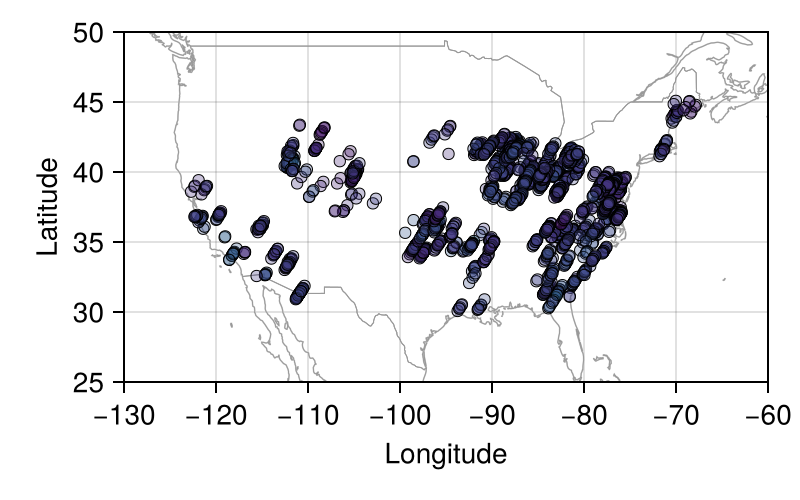}
        \caption*{May}
    \end{subfigure}

    \vspace{0.3cm} 

    \begin{subfigure}[b]{0.32\textwidth}
        \centering
        \includegraphics[width=\textwidth]{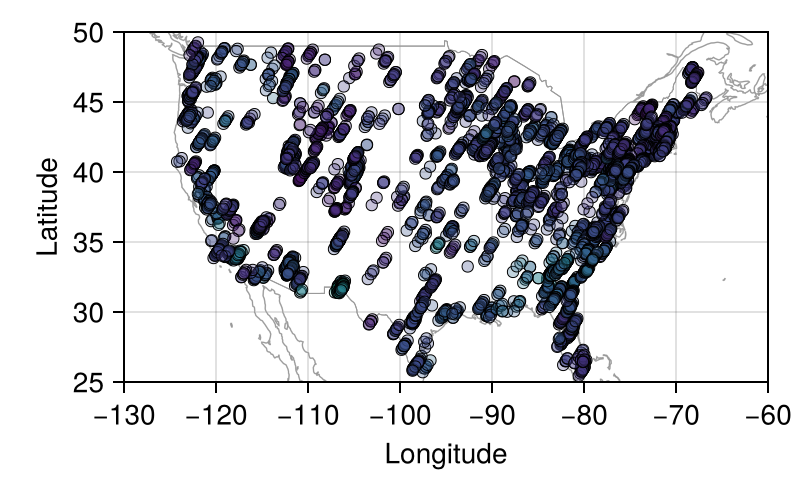}
        \caption*{February}
    \end{subfigure}
    \hfill
    \begin{subfigure}[b]{0.32\textwidth}
        \centering
        \includegraphics[width=\textwidth]{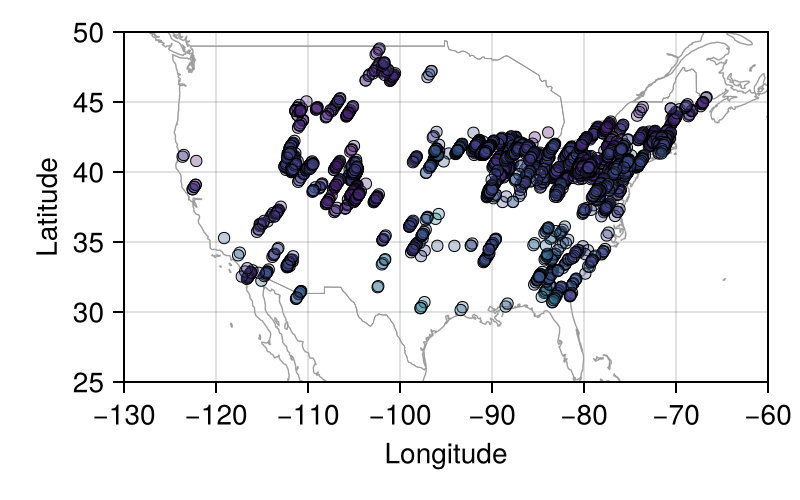}
        \caption*{April}
    \end{subfigure}
    \hfill
    \begin{subfigure}[b]{0.32\textwidth}
        \centering
        \includegraphics[width=\textwidth]{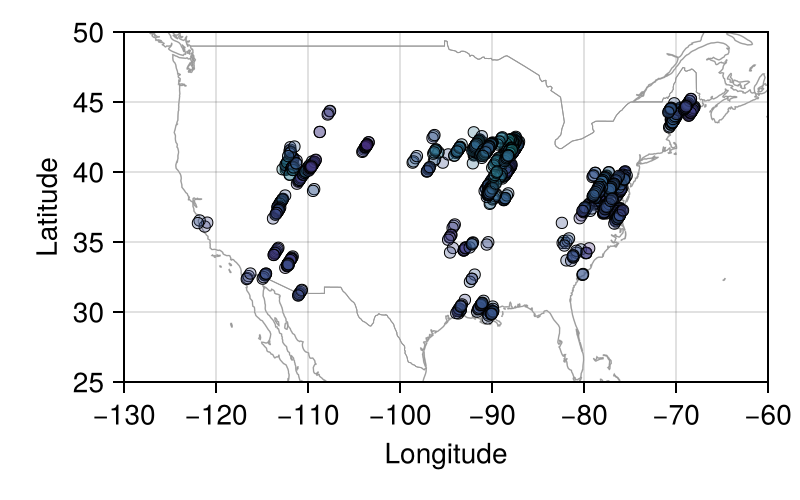}
        \caption*{June}
    \end{subfigure}

    \vspace{0.5cm}

    \begin{subfigure}[b]{0.32\textwidth}
        \centering
    \end{subfigure}
    \hfill
    \begin{subfigure}[b]{0.32\textwidth}
        \centering
    \end{subfigure}
    \hfill
    \begin{subfigure}[b]{0.32\textwidth}
        \centering
    \end{subfigure}

    \caption{
    Spatial distribution of rejected PM$_{2.5}$ measurements across the continental U.S.\ from January to June 2025. 
    Each panel shows locations of rejections overlaid on observed monitoring sites (gray). 
    The panels in each column correspond to consecutive months (Jan–Feb, Mar–Apr, May–Jun), illustrating seasonal progression in detected pollution events.
    }
    \label{fig:air_pollution_space}
\end{figure}

\paragraph{Interpretation and implications.}
Combining temporal and spatial analyses demonstrates that our procedure can adapt flexibly to different covariate structures.
Even when the underlying distribution of $Y$ is asymmetric, an appropriate transformation (e.g., log) restores conditional symmetry, allowing the estimation of $m(x)$ and valid FDR control.
Moreover, results from the two-dimensional spatial covariate setting show that our method naturally scales beyond one-dimensional covariates.
When sufficient local data are available, it reliably identifies region-specific pollution events without any explicit parametric modeling of spatial dependence.

In summary, these results highlight two key advantages of our framework:
\begin{enumerate}
\item[(1)] robust applicability to non-symmetric data after transformation, and 
\item[(2)] capacity to incorporate multiple covariates in a data-adaptive manner. 
\end{enumerate}
Thus, our method provides a general-purpose tool for discovering structured anomalies while maintaining rigorous false discovery control.

\section{Concluding Remarks} 
\label{sec:conclusion}

In this paper, we proposed a novel method for multiple hypothesis testing that directly operates at the level of raw data, rather than relying on precomputed or covariate-adjusted p-values. Our method is built upon two key assumptions: the conditional symmetry assumption, which assumes that the null distribution is symmetric around a covariate-dependent center, and the zero assumption, which ensures identifiability between null and alternative components.
Under these assumptions, we developed a data-trimming algorithm to estimate the conditional null center  in a nonparametric and covariate-adaptive manner. Using this estimated center, we derived empirical $p$-values from the observed data and subsequently optimized a covariate-dependent rejection threshold via a neural network guided by an augmented Lagrangian optimization scheme. This combination allows our procedure to maximize the number of discoveries while rigorously maintaining false discovery rate (FDR) control.

A main contribution of our work is that it bridges the gap between raw-data modeling and p-value–based inference: by estimating the center of the conditional null distribution, our approach enables the generation of valid, covariate-adjusted p-values from first principles. This makes existing FDR-controlling methods that assume given p-values applicable within our framework. In this sense, our method not only generalizes the existing approaches but also provides a practical foundation for applying them when valid covariate-adjusted p-values are otherwise unavailable.

Extensive simulation studies demonstrated that the proposed method achieves superior power while maintaining nominal FDR levels across various covariate structures. Real-data analyses on age-specific blood pressure and U.S. air pollution datasets further confirmed its robustness and interpretability, showing that it can flexibly handle covariate-dependent nulls and data transformations to meet the symmetry assumption.

Overall, our proposed method offers a unified, data-driven approach to FDR control that is broadly applicable to real-world problems where the null distribution depends on covariates and parametric assumptions are difficult to justify. Future research may extend this framework to high-dimensional covariates, use kernel estimators in the trimming procedure, explore theoretical guarantees under weaker assumptions, and {investigate its integration with modern deep-learning architectures for large-scale scientific discovery }.

\section*{Acknowledgement}
This work was supported by the National Research Foundation of Korea (NRF) grant
funded by the Korea government (MSIT) (RS-2025-00556575).

\appendix

\section{Theoretical Analysis}
\label{sec:theory}
In this section, we justify the conditions required for the {\bf Algorithm} \ref{alg:alg1} to properly converge to the mean 
under the null hypothesis.

From ~\citet{miao2006sym}, we know that the test statistic $T$ follows an asymptotic normality under the following conditions in {\bf Theorem} \ref{thm:cond}. Distributions such as  $t$-distribution with degrees of freedom 3, 
Gaussian, logisitic, and double exponential distribution have been verified to satisfy such conditions.

\begin{theorem}[{\cite[Corollary~1]{miao2006sym}}]\label{thm:cond}
Let $F$ be a symmetric distribution satisfying the following regularity conditions:
\begin{enumerate}[(C1)]
    \item The quantile function $F^{-1}$ is continuous on $(0,1)$ and locally Lipschitz on every subinterval bounded away from $\{0,1\}$; any discontinuities occur on a bounded measurable set of Lebesgue measure zero.
    \item $f(\nu) \neq 0$, and the Riemann integral $\int_0^1 (F^{-1})'(u)\,[u(1-u)]^{1/2}\,du$ converges absolutely.
    \item For $i=1,2$, there exist $\delta_i \in (0,1)$ such that for all $k_i>0$ there exists $M_i<\infty$ satisfying:
    \begin{align}
        &\text{if } 0<u_1,u_2<\delta_1 \text{ and } k_1^{-1}<\tfrac{u_1}{u_2}<k_1, 
        &&M_1^{-1}<\tfrac{(F^{-1})'(u_1)}{(F^{-1})'(u_2)}<M_1, \\
        &\text{if } 1-\delta_2<u_1,u_2<1 \text{ and } k_2^{-1}<\tfrac{1-u_1}{1-u_2}<k_2, 
        &&M_2^{-1}<\tfrac{(F^{-1})'(u_1)}{(F^{-1})'(u_2)}<M_2.
    \end{align}
\end{enumerate}
Then,
\[
\sqrt{n}\left(T-\frac{\lambda}{\sqrt{\pi/2}\,\tau}\right)
\xrightarrow{d} N(0, \sigma_T^2).
\]
\end{theorem}

Next, we justify that calculating the test statistic for the truncated datasets is reasonable.
\begin{theorem}
    If a symmetric distribution satisfies the conditions of Theorem \ref{thm:cond}, all symmetrically truncated distributions also follow the conditions.
\end{theorem}

\begin{proof}
    Define the cumulative distribution function(CDF) $G$, the probability density function(PDF) $g$, and the mean $m$. Then, the PDF of the truncated distribution $g_t$ is 
    \begin{equation}
        g_t(y) = \frac{g(y)}{\int_{a}^{2m-a}g(y)} I_{[a, 2m-a]}(y)= \frac{g(y)}{G(2m-a) - G(a)}I_{[a, 2m-a]}(y),
    \end{equation}
    
    and the CDF $G_t$ is 
    \begin{equation}
        G_t(y) = \frac{G(y) - G(a)}{G(2m-a) - G(a)}.
    \end{equation}
    Thus, we obtain 
    \begin{equation}
        G_t^{-1}(u) = G^{-1}(u (G(2m-a) - G(a)) + G(a)) =: G^{-1}(l(u))
    \end{equation} where $l(u) = (G(2m-a) - G(a)) u+ G(a)$ is a linear function of $u$.
    
    Since $G_t^{-1}$ is the composition of $G^{-1}$ and $l$, $G_t^{-1}$  satisfies the first part of Condition 1. Also, the boundary of the interval $[a, 2m-a]$ has Lebesgue measure zero, so
    $\{x:G_t^{-1}(x) \text{ exists and is continuous} \} = \{x:G^{-1}(x) \text{ exists and is continuous} \} \cap [a, 2m-a]$
    satisfies the second part of Condition 1.
    
    Next, we have $G_t(\nu) = \frac{g(\nu)}{G(2m-a) - G(a)} I_{[a, 2m-a]}(y) \neq 0$,
    and $\int_0^1 (G_t^{-1})'(u)[u(1-u)]^{\frac{1}{2}}du$ exists since the support of $G_t$ is bounded, so condition 2 holds for $G_t$.
    
    Finally, 
    \begin{equation}
        (G_t^{-1})'(u) = \frac{1}{G_t'(G_t^{-1}(u))} = \frac{G(2m-a) - G(a)}{G'(G^{-1}(l(u)))} = (G(2m-a) - G(a))(G^{-1})'(l(u)),
    \end{equation}
    and under the condition $0 < u_1, u_2 < \delta_1$ and $k_1^{-1} < \frac{u_1}{u_2} < k_1$ 
    for $k_1>1$,
    \begin{equation}
        \frac{l(u_1)}{l(u_2)} = \frac{(G(2m-a) - G(a)) u_1 + G(a)}{(G(2m-a) - G(a)) u_2 + G(a)} < \frac{(G(2m-a) - G(a)) k_1 u_2 + k_1 G(a)}{(G(2m-a) - G(a)) u_2 + G(a)} = k_1,
    \end{equation}
    \begin{equation}
        \frac{l(u_1)}{l(u_2)} = \frac{(G(2m-a) - G(a)) u_1 + G(a)}{(G(2m-a) - G(a)) u_2 + G(a)} > \frac{(G(2m-a) - G(a)) u_2/k_1 + G(a)/k_1}{(G(2m-a) - G(a)) u_2 + G(a)} = \frac{1}{k_1}
    \end{equation}
    and similar for the conditions under condition(3-2).
\end{proof}

\textbf{Theorem} \ref{thm:trim} gives an insight of the \bf Algorithm} \ref{alg:alg1} which makes 
the remaining data to be symmetric after truncations. 
More specifically, 
when the distribution of 
remaining data after truncation is 
still skewed, the {\bf Algorithm} \ref{alg:alg1} has very little chance to 
truncate the data in the wrong side so that 
the {\bf Algorithm} \ref{alg:alg1} can have more chance to truncate the data  in the long tail or stop truncation.   
\begin{theorem}
    If $\mu > \nu$, the procedure does not trim the smaller side(the wrong side) with probability tending to 1 as $n \rightarrow \infty$. For the case of $\mu < \nu$, we have a similar result that 
    the procedure does not trim the larger side with probability tending to 1 as $n \rightarrow \infty$. 
    \label{thm:trim}
\end{theorem}

\begin{proof}
    First, observe that 
    \begin{align}
        & \mathbb{P}(\mbox{ The minimum value is trimmed } ) \\
        = \ & \mathbb{P}\left( \frac{\hat{\mu}_n - \hat{\nu}_n}{\hat{d_n}} < Z_{1 - \frac{\alpha}{2}} \right) \\
        = \ & \mathbb{P}\left( \sqrt{n}(T_n - \frac{\mu - \nu}{d}) < \sqrt{n} \left( Z_{1-\frac{\alpha}{2}} - \frac{\mu - \nu}{d} \right) \right) \\
        \leq \ & \mathbb{P}(\sqrt{n}(T_n - \frac{\mu - \nu}{d}) < \sqrt{n} Z_{1-\frac{\alpha}{2}} )
    \end{align}
    where $Z_{\alpha} $ satisfies $\mathbb{P}(Z>Z_\alpha) = \alpha$ with $Z$ the standard normal distribution. 
    Since $\sqrt{n}(T_n - \frac{\mu - \nu}{d})$ converges to $N(0,\sigma_T^2)$ in distribution and $\sqrt{n} C_{1-\frac{\alpha}{2}}$ diverges to $-\infty$, we conclude that the procedure trims the smallest value with probability tending to 0 as $n \rightarrow \infty$.
\end{proof}

Finally, we provide a sufficient condition for the
{\bf Algorithm} \ref{alg:alg1} not to terminate 
when the distribution of remaining data 
is not symmetric enough yet. 

Our baseline idea is that 
there exists a distribution 
$G_0$ such that the center of 
$G_0$ is $\mu_0=m(x)$ due to the zero assumption in Assumption {\bf A2}. 
While we repeat truncation of the data, 
the remaining data may not be symmetric, so we 
consider two distributions in the left and right tails 
which explain such asymmetric distribution of the whole remaining data. 
We introduce those two distributions denoted by 
$G_1$ and $G_2$.

Summarizing this, we can consider the remaining data as an i.i.d sample from the mixture model $G = \pi_0G_0 + \pi_lG_l + \pi_rG_r$ 
where $\pi_0$, $\pi_l$ and $\pi_r$ are 
the mixing proportions. Intuitively, the following {\bf Theorem} \ref{thm:suff} indicates that if the asymmetric parts of both tails have significantly different means or proportions, the {\bf algorithm} \ref{alg:alg1} will not stop trimming. For example, if $\pi_l = \pi_r$ and $\mu_r - \mu_0 > \mu_0 - \mu_l$ (that is, the two distributions in the tail region have the same proportion and $G_l$ is located further from $G_0$ than $G_r$ is), then trimming occurs from $G_l$ with probability tending to 1 as $n$ increases 
while the distribution of the remaining data is not yet symmetric enough.
\begin{theorem}
    Suppose the trimmed distribution can be expressed as 
    $G = \pi_0G_0 + \pi_lG_l + \pi_rG_r$ where $G_0$ is symmetric with respect to $\mu_0$, $G_l$ is supported at $(-\infty, \mu_0)$ with mean $\mu_l$, and $G_r$ is supported at $(\mu_0, \infty)$ with mean $\mu_r$, and $\pi_0 + \pi_l + \pi_r = 1$. If $\pi_l > \pi_r$ and either of the two conditions (i) $\pi_l( \frac{L}{2\pi_0} - (\mu_0 - \mu_l) ) < \pi_r(\frac{L}{2\pi_0} - (\mu_r - \mu_0))$ or 
    (ii) $\pi_l(\mu_0 - \mu_l) < \pi_r(\mu_r - \mu_0)$ is satisfied, {\bf Algorithm} \ref{alg:alg1} does not stop at the moment with probability tending to 1 as $n \rightarrow \infty$. Here, $L$ is the Lipschitz constant of $G_0^{-1}$ on $(0,1)$ (The opposite is similar for $\pi_l < \pi_r$.)
    \label{thm:suff}
\end{theorem}

\begin{proof}
    Define $\mu,\nu$ as the mean, median of $G$, respectively. $G_0^{-1}$ is an increasing function, and it can be easily shown that $\frac{1}{2}> \frac{1}{2\pi_0} - \frac{\pi_l}{\pi_0}$ through a simple calculation. Asymptotically, the sign of $\mu - \nu$ determines the side on which the trimming occurs. First, from the definition of the Lipschitz constant and \textbf{Theorem} \ref{thm:trim}, we conclude that condition (i) is a sufficient condition for the maximum value to be trimmed:
    \begin{align*}
        & \frac{\pi_l(\mu_0 - \mu_l) - \pi_r(\mu_r - \mu_0)}{\frac{1}{2} - (\frac{1}{2\pi_0} - \frac{\pi_l}{\pi_0})} > L > \frac{G_0^{-1}(\frac{1}{2}) - G_0^{-1}(\frac{1}{2\pi_0} - \frac{\pi_l}{\pi_0})}{\frac{1}{2} - (\frac{1}{2\pi_0} - \frac{\pi_l}{\pi_0})} \\
        \Rightarrow \ & \pi_l(\mu_l - \mu_0) + \pi_r(\mu_r - \mu_0) < G_0^{-1}(\frac{1}{2\pi_0} - \frac{\pi_l}{\pi_0}) - G_0^{-1}(\frac{1}{2})\\
        \Rightarrow \ & \pi_0\mu_0 + \pi_l\mu_l + \pi_r\mu_r < G_0^{-1}(\frac{1}{2\pi_0} - \frac{\pi_l}{\pi_0}) \\
        \Rightarrow \ & \mu < \nu
    \end{align*}
    Similarly, a sufficient condition for the minimum value to be trimmed is given as condition (ii), justified by the following equations.
    \begin{align*}
        & \pi_l(\mu_0 - \mu_l) - \pi_r(\mu_r - \mu_0) < 0 < G_0^{-1}(\frac{1}{2}) - G_0^{-1}(\frac{1}{2\pi_0} - \frac{\pi_l}{\pi_0}) \\
        \Rightarrow \ & \pi_0\mu_0 + \pi_l\mu_l + \pi_r\mu_r > G_0^{-1}(\frac{1}{2\pi_0} - \frac{\pi_l}{\pi_0}) \\
        \Rightarrow \ & \mu > \nu
    \end{align*}
    
\end{proof}




\bibliographystyle{plainnat}
\bibliography{reference}

@misc{AgeBP, 
  author       = {{National Center for Health Statistics}},
  title        = {{National Health and Nutrition Examination Survey: August 2021– August 2023}},
  howpublished = {\url{https://www.cdc.gov/nchs/nhanes/}},
  year         = {2023},
  note         = {Accessed: August 2023}
}

@article{benjamini1995controlling,
  title={Controlling the false discovery rate: a practical and powerful approach to multiple testing},
  author={Benjamini, Yoav and Hochberg, Yosef},
  journal={Journal of the Royal Statistical Society: Series B (Methodological)},
  volume={57},
  number={1},
  pages={289--300},
  year={1995}
}

@article{ignatiadis2016data,
  title={Data-driven hypothesis weighting increases detection power in genome-scale multiple testing},
  author={Ignatiadis, Nikolaos and Klaus, Bernd and Zaugg, Judith B and Huber, Wolfgang},
  journal={Nature Methods},
  volume={13},
  number={7},
  pages={577--580},
  year={2016}
}

@article{lei2018adapt,
  title={AdaPT: An interactive procedure for multiple testing with side information},
  author={Lei, Lihua and Fithian, William},
  journal={Journal of the Royal Statistical Society: Series B (Statistical Methodology)},
  volume={80},
  number={4},
  pages={649--679},
  year={2018}
}

@article{boca2018regression,
  title={A regression framework for the proportion of true null hypotheses},
  author={Boca, Simina M and Leek, Jeffrey T},
  journal={Biostatistics},
  volume={19},
  number={2},
  pages={228--241},
  year={2018}
}

@inproceedings{xia2020neuralfdr,
  title={NeuralFDR: Learning discovery thresholds from hypothesis features},
  author={Xia, Yanjun and Mukherjee, Rajarshi and Sarkar, Purnamrita},
  booktitle={International Conference on Machine Learning},
  pages={10514--10524},
  year={2020},
  organization={PMLR}
}

@inbook{miao2006sym,
  title={A new test of symmetry about an unknown median},
  author={Miao, Weiwen and Gel, Yulia R and Gastwirth, Joseph L},
  booktitle={Random walk, sequential analysis and related topics: A festschrift in honor of Yuan-Shih Chow},
  pages={199--214},
  year={2006},
  publisher={World Scientific}
}

@inproceedings{tancik2020fourier,
  title = {Fourier Features Let Networks Learn High Frequency Functions in Low Dimensional Domains},
  author = {Tancik, Matthew and Srinivasan, Pratul P. and Mildenhall, Ben and Fridovich-Keil, Sara and Raghavan, Nithin and Singhal, Utkarsh and Ramamoorthi, Ravi and Barron, Jonathan T. and Ng, Ren},
  booktitle = {Advances in Neural Information Processing Systems (NeurIPS)},
  year = {2020},
  url = {https://arxiv.org/abs/2006.10739}
}

@misc{EPA_AirData,
  author       = {{U.S. Environmental Protection Agency}},
  title        = {AirData: Download Files},
  year         = {2025},
  howpublished = {\url{https://aqs.epa.gov/aqsweb/airdata/download_files.html}},
  note         = {Accessed September 2025}
}

@article{xing2016impact,
  title={The impact of PM2. 5 on the human respiratory system},
  author={Xing, Yu-Fei and Xu, Yue-Hua and Shi, Min-Hua and Lian, Yi-Xin},
  journal={Journal of thoracic disease},
  volume={8},
  number={1},
  pages={E69},
  year={2016}
}

@book{nocedal2006numerical,
  title={Numerical Optimization},
  author={Nocedal, Jorge and Wright, Stephen J.},
  year={2006},
  edition={2nd},
  publisher={Springer},
  address={New York},
  isbn={9780387400655}
}

@article{efron2004large,
  title        = {Large-scale simultaneous hypothesis testing: The choice of a null hypothesis},
  author       = {Efron, Bradley},
  journal      = {Journal of the American Statistical Association},
  volume       = {99},
  number       = {465},
  pages        = {96--104},
  year         = {2004},
  publisher    = {Taylor \& Francis},
  doi          = {10.1198/016214504000000089}
}

@article{Dai2022FDRDataSplitting,
  title        = {False Discovery Rate Control via Data Splitting},
  author       = {Dai, Chenguang and Lin, Buyu and Xing, Xin and Liu, Jun S.},
  journal      = {Journal of the American Statistical Association},
  volume       = {118},
  number       = {544},
  pages        = {2503--2520},
  year         = {2022},
  doi          = {10.1080/01621459.2022.2060113},
  url          = {https://arxiv.org/abs/2002.08542}
}

@article{yano2018hypertension,
  title={Hypertension and Cardiovascular Risk in Young Adults: The CARDIA Study},
  author={Yano, Yuichiro and Tanner, Robin M and Sakaniwa, Rico and Min, Yumi I and Wei, Gao and Muntner, Paul and Viera, Anthony J and Vasan, Ramachandran S and Shimbo, Daichi},
  journal={JAMA},
  volume={320},
  number={17},
  pages={1775--1785},
  year={2018},
  publisher={American Medical Association}
}

@article{pletcher2008prehypertension,
  title={Prehypertension during Young Adulthood and Coronary Calcium Later in Life: The CARDIA Study},
  author={Pletcher, Mark J and Bibbins-Domingo, Kirsten and Liu, Kiang and Sidney, Stephen and Havlik, Richard and Hulley, Stephen B},
  journal={Annals of Internal Medicine},
  volume={149},
  number={2},
  pages={91--99},
  year={2008},
  publisher={American College of Physicians}
}

@article{chen2008tracking,
  title={Tracking of blood pressure from childhood to adulthood: a systematic review and meta-regression analysis},
  author={Chen, Xiaoli and Wang, Youfa},
  journal={Circulation},
  volume={117},
  number={25},
  pages={3114--3114},
  year={2008},
  publisher={LWW}
}

@article{min2015aggressive,
  title={Aggressive Antihypertensive Drug Use Can Harm Elderly},
  author={Min, Lillian and Kim, Chi-Man and Rich, Michael and Bell, Stephen and Tinetti, Mary},
  journal={Medscape},
  year={2015},
  url={https://www.medscape.com/viewarticle/845347}
}

@article{mallery2025promoting,
  title={Promoting higher blood pressure targets for frail older adults: A consensus guideline from Canada},
  author={Mallery, Leisha and Van Der Straeten, Jonathan and Bherer, Louis and Chertkow, Howard and Davis, Julie and Demers, Charles and Gravel, Simon and Hogan, David B and Lacasse, Yves and Levesque, Luc and others},
  journal={Canadian Journal of Cardiology},
  volume={41},
  number={8},
  pages={983--991},
  year={2025},
  publisher={Elsevier}
}

@article{arguedas2013aggressive,
  title={Aggressive blood pressure lowering is dangerous: the J-curve},
  author={Arguedas, Julio A and Perez, Gonzalo and Romero-Montero, Alexandra},
  journal={Hypertension},
  volume={62},
  number={1},
  pages={1--2},
  year={2013},
  publisher={LWW}
}

@article{zhang2019fast,
  author    = {Martin J. Zhang and Fei Xia and James Zou},
  title     = {Fast and covariate-adaptive method amplifies detection power in large-scale multiple hypothesis testing},
  journal   = {Nature Communications},
  year      = {2019},
  volume    = {10},
  number    = {1},
  pages     = {11247},
  doi       = {10.1038/s41467-019-11247-0}
}
\end{document}